\documentclass[12pt,reqno]{amsart}
\usepackage[hmargin=0.7in,height=9.5in]{geometry}
\usepackage{amssymb,amsthm, times, enumerate}
\usepackage{delarray,verbatim}
\usepackage{subfigure}

\usepackage{ifpdf}
\ifpdf
\usepackage[pdftex]{graphicx}
 \else
\usepackage[dvips]{graphicx}
 \fi

\usepackage{ifpdf}
\usepackage{color}
\definecolor{webgreen}{rgb}{0,.5,0}
\definecolor{webbrown}{rgb}{.8,0,0}
\definecolor{emphcolor}{rgb}{0.95,0.95,0.95}

\usepackage{hyperref}
\hypersetup{%
          colorlinks=true,
          linkcolor=webbrown,
          filecolor=webbrown,
          citecolor=webgreen,
          breaklinks=true}
\ifpdf \hypersetup{pdftex,
            pdfstartview=FitH, 
            bookmarksopen=true,
            bookmarksnumbered=true
} \else \hypersetup{dvips} \fi

\newcommand {\costb}{\gamma_b}
\newcommand {\costs}{\gamma_s}
\newcommand {\ph}{\hat{p}}
\newcommand {\ah}{\hat{\alpha}}
\newcommand {\pcheck}{\tilde{p}}
\newcommand {\acheck}{\tilde{\alpha}}

\renewcommand{\S}{\mathcal{S}}

\numberwithin{equation}{section}

\newtheorem{theorem}{Theorem}[section]
\newtheorem{proposition}{Proposition}[section]

\newtheorem{corollary}{Corollary}[section]
\newtheorem{remark}{Remark}[section]
\newtheorem{lemma}{Lemma}[section]

\newtheorem{assumption}{Assumption}[section]

\numberwithin{remark}{section} \numberwithin{proposition}{section}
\numberwithin{corollary}{section}
\newcommand {\R}{\mathbb{R}}

\newcommand {\F}{\mathcal{F}}

\newcommand {\p}{\mathbb{P}}

\newcommand {\E}{\mathbb{E}}

\def\1{1{\hskip -3.3 pt}\hbox{I}}

\newcommand{\diff}{{\rm d}}

\newcommand{\lev}{L\'{e}vy }

\newcommand {\lap}{\zeta}
\newcommand{\lapinv}{\Phi(r)}

\title[Default Swap Games]{Default Swap Games Driven by Spectrally Negative L\'{E}vy  Processes$^*$}\thanks{$*$\,This draft: \today.  }
\author[M. Egami]{Masahiko Egami$^\diamond$}\thanks{$\diamond$\,Graduate School of Economics,
Kyoto University, Sakyo-Ku, Kyoto, 606-8501, Japan. Email: \mbox{\em
egami@econ.kyoto-u.ac.jp}}
\author[T. Leung]{\,\,Tim Leung$^\dag$\,}\thanks{$\dag$\, IEOR Department, Columbia University, New York NY 10027, USA. Email: \mbox{{\em leung@ieor.columbia.edu}}}
\author[K. Yamazaki]{\,Kazutoshi Yamazaki$^\ddag$}\thanks{$\ddag$\, (corresponding author) Center for the Study of Finance and Insurance,
Osaka University, 1-3 Machikaneyama-cho, Toyonaka City, Osaka 560-8531, Japan. Email: \mbox{{\em
k-yamazaki@sigmath.es.osaka-u.ac.jp}}. Phone: +81-(0)6-6850-6469.  Fax: +81-(0)6-6850-6092}

\begin{document}
\begin{abstract}
This paper studies game-type credit default swaps that allow the protection buyer and seller to raise or reduce their
respective positions once prior to default. This leads to the study of
an optimal stopping game subject to early default
termination. Under a structural credit risk model based on spectrally
negative \lev processes, we apply the principles of smooth and
continuous fit to identify the equilibrium exercise strategies for the buyer and the  seller. We then rigorously prove the existence of the Nash equilibrium and compute the contract value at equilibrium.  Numerical examples are
provided to illustrate the impacts of default risk and other contractual
features on the players' exercise timing at equilibrium.
\end{abstract}
\maketitle
{\noindent \small{\textbf{Keywords:}\,  optimal stopping games; Nash equilibrium; \lev processes; scale function; credit default swaps}\\
 \noindent \small{\textbf{JEL Classification:}\, C73, G13,  G33, D81}\\
\noindent \small{\textbf{Mathematics Subject Classification (2010):}\, 91A15, 60G40, 60G51, 91B25  }}

\section{Introduction}\label{section_intro}
Credit default swaps (CDSs) are among the most liquid and widely used
credit derivatives for trading and managing default risks. Under a vanilla
CDS contract, the protection buyer pays a periodic premium  to the
protection seller in exchange for a payment if the reference entity defaults
before expiration. In order to control the credit risk exposure, investors can adjust the premium and notional amount prior to default by appropriately combining a market-traded default swaption with a vanilla CDS position, or use the over-the-counter traded products such as the callable CDSs (see \cite[Chapter 21]{Brigobook}). In a recent related work \cite{Leung_Yamazaki_2010},  we studied  the \emph{optimal timing} to step up or down a CDS position under a general \lev credit risk model.


The current paper studies the game-type CDSs that allow both the protection
buyer and seller to change the swap position once prior to default.
Specifically, in the step-up (resp.\ step-down) default swap game, as soon as
the buyer or the seller, whoever first, exercises prior to default, the
notional amount and premium  will be increased (resp.\ decreased) to a
pre-specified level upon exercise. From the exercise time till default, the
buyer will pay the new premium and the seller is subject to the new default
liability. Hence, for a given set of contract parameters, the buyer's objective
is  to maximize the expected net cash flow while the seller wants to
minimize it, giving rise to a two-player optimal stopping game.

We   model  the default time as the \emph{first passage time} of a  general
\emph{exponential \lev process}  representing some underlying
asset value.  The default event occurs  either when the underlying asset value moves continuously to the lower default barrier, or when it jumps below the default barrier.
This is an extension of  the original structural  credit risk
approach introduced by Black and Cox \cite{BlackCox76}  where the asset
value follows a geometric Brownian motion. As is well known \cite{Duffie2003}, the incorporation of unpredictable  jump-to-default is useful  for explaining a number of market observations, such as the non-zero short-term limit of credit spreads.  Other related credit risk models based on \lev and other jump processes include \cite{schoutensCDS07,helberenkrogers,zhou2001}.


 The default swap game is formulated  as a variation of the standard
 optimal stopping games in the literature (see, among others,
  \cite{dynkin_book1968, Peskir_2008} and references therein).  However, while typical
 optimal stopping games end at the  time of exercise by either player,   the
  exercise time in the default swap game does not terminate the
 contract, but merely alters the premium
forward and the future protection amount to be paid at default time. In
fact,  since default may arrive before either party exercises, the game may
be terminated early \emph{involuntarily}.

The central challenge of  the  default swap games  lies in  determining the pair
of stopping times that yield the \emph{Nash equilibrium}.  Under a structural credit
risk model based on \emph{spectrally negative} \lev processes, we analyze and
calculate the equilibrium exercise strategies for the protection buyer and
seller. In addition, we  determine the equilibrium premium of the default swap
game so that the   expected discounted cash  flows for the two parties
coincide at contract inception.

 Our solution approach starts with a decomposition of the default swap game
 into a combination of a perpetual  CDS and  an optimal stopping game with
   early  termination from default. Moreover, we utilize a  symmetry between the step-up and step-down games, which
   significantly simplifies our analysis as it is sufficient to study either  case.
For a general spectrally negative \lev process (with a non-atomic \lev measure), we provide the conditions for the existence of the Nash
equilibrium. Moreover, we  derive the buyer's and seller's optimal threshold-type exercise strategies using the principle of continuous and smooth fit, followed by a rigorous verification theorem via martingale arguments.

For our analysis of the game equilibrium,  the \emph{scale
function} and a number of fluctuation identities of spectrally negative \lev
processes are particularly useful. Using our analytic results, we   provide a bisection-based algorithm
for the efficient computation of the buyer's and seller's exercise thresholds as well as the equilibrium premium, illustrated in a series of numerical examples. Other recent applications of
spectrally negative \lev processes include derivative pricing
\cite{Alili2005,Avram_2004},   optimal dividend problem
\cite{Avram_et_al_2007,Kyprianou_Palmowski_2007,Loeffen_2008}, and
capital reinforcement timing \cite{Egami_Yamazaki_2010}. We   refer the
reader to \cite{Kyprianou_2006} for  a comprehensive account.


To our best knowledge, the step-up and step-down default
swap games and the associated optimal stopping games have
not been studied elsewhere. There are a few related studies on stochastic
games driven by spectrally negative or positive \lev processes; see e.g.\
\cite{BaurdouxKypri08} and \cite{BauKyrianouPardo2011}.  For optimal
stopping games driven by a strong  Markov process, we refer to the recent
papers by  \cite{Peskir_2008} and \cite{Peskir_2009}, which study the
existence and mathematical characterization of Nash  equilibrium. Other  game-type derivatives in the literature include
Israeli/game options \cite{kifer2000, Kyprianou_israeli}, defaultable game options
\cite{BieleckiCrepeyJR08}, and convertible bonds
\cite{kallsenkuhn05,sirbushreve_convert}.


The rest of the paper is organized as follows. In Section \ref{section_model}, we formulate
the default swap game under a general \lev model. In Section
\ref{section_spec_neg_lev}, we focus on the spectrally negative \lev model and analyze the Nash equilibrium. Section \ref{section_numer} provides the numerical study of the default swap games for the case with i.i.d.\ exponential jumps. Section \ref{section_conclude} concludes the paper and presents some ideas for future work.
All  proofs are given in the Appendix.

\section{Game Formulation} \label{section_model}

On a complete probability space  $(\Omega, \F, \p)$, we assume there exists a L\'{e}vy process $X=\{X_t;\, t\ge 0\}$ and denote by $\mathbb{F}=(\F_t)_{t\ge 0}$ the filtration generated by $X$.  The value of the reference entity (a company stock or other assets) is assumed to evolve according to an \emph{exponential \lev process} $S_t = e^{X_t}$, $t\geq0$. Following the Black-Cox \cite{BlackCox76} structural approach, the default event is triggered by $S$ crossing a lower level $D$. Without loss of generality, we can take $\log D =0$ by shifting the initial value $x\in\R$. Henceforth, we shall work with the default time
\[{\sigma_0}:=\inf\{\,t\ge 0\,:\, X_t\,  \leq \,0\,\},\]
where $\inf \emptyset\! =\! \infty$ by convention.  We denote by  $\p^x$ the probability law and $\E^x$ the expectation with $X_0=x$.


We consider a default swap contract that gives the protection buyer and seller an option to change the premium and notional amount before default  for a fee, whoever exercises first.  Specifically, the  buyer begins by paying premium at rate $p$ over time  for a notional amount $\alpha$ to be paid at default.  Prior to default, the buyer and the seller can select a time to switch to a new premium $\ph$ and notional amount $\ah$. When the buyer exercises, she is incurred the fee  $\costb$ to be paid to the seller; when the seller exercises, she is incurred $\costs$ to be paid to the buyer.  If the
buyer and the seller exercise simultaneously, then  both parties pay the fee upon exercise. We assume that $p$, $\ph$, $\alpha$, $\ah$, $\costb$,
$\costs \geq 0$ (see also Remark \ref{remark_zero} below).

Let $\S : = \left\{ \tau \in \mathbb{F}: \tau \leq {\sigma_0} \; \text{ a.s.
}\right\}$  be the set of all stopping times \emph{smaller than or equal to}
the default time.  Denote the buyer's candidate exercise time by $\tau \in
\S$ and seller's candidate exercise time by $\sigma \in \S$, and let $r>0$ be
the positive risk-free interest rate. Given any pair of exercise times
$(\sigma, \tau)$, the expected cash flow to the buyer is given by
\begin{multline}
V(x; \sigma, \tau) :=\E^x \left[ -\int_0^{\tau \wedge \sigma} e^{-rt} p\,\diff t + 1_{\{\tau \wedge \sigma < \infty\}}\bigg( -\int_{\tau \wedge \sigma}^{\sigma_0} e^{-rt}\ph\,\diff t \right. \\ \left. + e^{-r {\sigma_0}}(\ah 1_{\{{\tau \wedge \sigma} <{\sigma_0}\}}   +\alpha 1_{\{{\tau \wedge \sigma} ={\sigma_0} \}}) + 1_{\{\tau \wedge \sigma < {\sigma_0} \}} e^{-r (\tau \wedge \sigma)}  \left(- \costb 1_{\{\tau \leq \sigma \}} +  \costs 1_{\{\tau \geq \sigma \}}  \right) \bigg) \right]. \label{def_V}
\end{multline} To the seller, the contract value   is  $-V(x; \sigma, \tau)$.
Naturally, the buyer wants to \emph{maximize} $V$ over $\tau$ whereas
the seller wants to \emph{minimize} $V$ over $\sigma$, giving rise to a
two-player optimal stopping game.

This formulation covers default swap games with the following provisions:
\begin{enumerate}\item \emph{Step-up Game}: if $\ph> p$ and $\ah >\alpha$, then the buyer and the seller are allowed to \emph{increase} the notional amount once from $\alpha$ to $\ah$ and the premium rate from $p$ to $\ph$ by paying the fee $\costb$ (if the buyer exercises) or $\costs$ (if the seller exercises).
\item \emph{Step-down Game}: if $\ph < p$ and $\ah < \alpha$, then the buyer and the seller are allowed to \emph{decrease} the notional amount once from $\alpha$ to $\ah$ and the premium rate from $p$ to $\ph$ by paying the fee $\costb$ (if the buyer exercises) or $\costs$ (if the seller exercises).  When $\ph = \ah =0$, we obtain a \emph{cancellation game} which allows the buyer and the seller to terminate the contract early.
\end{enumerate}

Our primary objective is to determine the pair of stopping times $(\sigma^*, \tau^*)\subset \mathcal{S}$, called the \emph{saddle point}, that constitutes the \emph{Nash equilibrium}:
\begin{align}
V(x; \sigma^*, \tau) \leq V(x; \sigma^*, \tau^*)  \leq V(x; \sigma, \tau^*), \quad \forall \, \sigma, \tau \in \S. \label{saddle_pt}
\end{align}

\begin{remark}A related concept is the Stackelberg equilibrium, represented by the equality  $V^*(x) = V_*(x)$, where $V^*(x) := \inf_{\sigma \in \S} \sup_{\tau \in \S} V(x; \sigma, \tau)$ and $V_*(x) := \sup_{\tau \in \S} \inf_{\sigma \in \S}  V(x; \sigma, \tau)$. See e.g.\ \cite{Peskir_2008} and \cite{Peskir_2009}. These definitions imply that  $V^*(x) \geq V_*(x)$. The existence of the Nash equilibrium  \eqref{saddle_pt} will also yield the  Stackelberg equilibrium via the reverse inequality:
\begin{align*}
V^*(x) \leq \sup_{\tau \in \S} V(x; \sigma^*, \tau) \leq V(x;\sigma^*,\tau^*) \leq  \inf_{\sigma \in \S} V(x; \sigma, \tau^*)\leq V_*(x).
\end{align*}
Herein, we shall focus our analysis on the Nash equilibrium.  \end{remark}

Our main results on the Nash equilibrium are summarized in Theorems \ref{theorem_equilibrium}-\ref{theorem_equilibrium_zero} for the spectrally negative \lev case. As preparation, we begin our analysis with two useful observations, namely,  the decomposition of $V$ and the  symmetry between the step-up and step-down games.

\subsection{Decomposition and Symmetry} In standard  optimal stopping games, such as the well-known Dynkin game \cite{dynkin_book1968},  random payoffs are realized at either player's exercise time. However,  our default swap game is not terminated at the buyer's or seller's exercise time. In fact, upon exercise only the contract terms will change, and there will be a terminal transaction at default time. Since default may arrive before either party exercises the
 step-up/down option, the game may be terminated early \emph{involuntarily}. Therefore, we shall transform the value function $V$ into another optimal stopping game that is more amenable for analysis.

 First, we  define the value of a (perpetual) CDS with premium rate $p$ and notional amount $\alpha$  by
\begin{align}C(x;p,\alpha) &:=\E^x \left[ -\int_0^{{\sigma_0} } e^{-rt} p\, \diff t + \alpha\, e^{-r {\sigma_0}}     \right] =  \left(\frac{p}{r} +\alpha\right) \lap(x) -\frac{p}{r}, \quad x > 0, \label{cds}
\end{align}where
\begin{align}
 \lap(x) := \E^x \left[  e^{-r{\sigma_0}}\right], \quad x \in  \R, \label{zeta}
\end{align}is the Laplace transform of ${\sigma_0}$. Next, we extract this CDS value from  the value function $V$.
Let
\begin{align}\acheck := \alpha - \hat{\alpha} \quad \textrm{and} \quad \pcheck := p - \hat{p}.\label{ap}\end{align}
\begin{proposition}[decomposition]\label{prop-V} For every $\sigma, \tau \in \S$ and $x > 0$, the value function admits the decomposition
\begin{align*}V(x; \sigma, \tau) &= C(x;p,\alpha)+ v(x; \sigma, \tau),\end{align*}
where $v(x; \sigma, \tau)\equiv v(x; \sigma, \tau ;\pcheck, \acheck, \costb, \costs)$ is defined by
\begin{align}
v(x; \sigma, \tau ;\pcheck, \acheck, \costb, \costs):= \E^x \left[ e^{-r (\tau\wedge \sigma)} \left( h(X_{\tau}) 1_{\{\tau < \sigma \}} + g(X_{ \sigma}) 1_{\{\tau > \sigma \}} + f(X_{\tau }) 1_{\{\tau = \sigma \}} \right) 1_{\{\tau \wedge \sigma < \infty\}} \right], \label{definition_v}
\end{align}
 with \begin{align}\label{hx}h(x) &\equiv h(x; \pcheck, \acheck, \costb) :=  1_{\{x > 0 \}}  \Big[ \Big(\frac{\pcheck}{r} - \costb \Big) - \Big(\frac{\pcheck}{r} +\acheck\Big) \lap(x)\Big], \\
\label{gx}g(x) &\equiv g(x; \pcheck, \acheck, \costs) := 1_{\{x > 0 \}}  \Big[ \Big(\frac{\pcheck}{r} +\costs \Big) - \Big(\frac{\pcheck}{r} +\acheck\Big) \lap(x)\Big], \\
\label{fx}f(x) &\equiv f(x; \pcheck, \acheck, \costb,\costs) := 1_{\{x > 0 \}}  \Big[ \Big(\frac{\pcheck}{r} - \costb +\costs \Big) - \Big(\frac{\pcheck}{r} +\acheck\Big) \lap(x)\Big].
\end{align}
\end{proposition}

Comparing \eqref{cds} and \eqref{hx}, we see  that $h(x) =1_{\{x>0\}} (C(x;-\pcheck, -\acheck) - \costb)$, which means that the buyer receives the CDS value $C(x;-\pcheck, -\acheck)$ at the cost of $\costb$ if she exercises before the seller.  For the seller, the payoff of exercising before the buyer is $ - g(x) = 1_{\{x>0\}} (C(x;\pcheck, \acheck) - \costs)$. Hence, in both cases the fees  $\costb$ and $\costs$ can be viewed as   strike prices.

Since $C(x;p,\alpha)$ does not depend on $(\sigma, \tau)$, Proposition \ref{prop-V} implies that  finding the saddle point $(\sigma^*,\tau^*)$ for the Nash equilibrium in \eqref{saddle_pt} is equivalent to showing that
\begin{align}
v(x;\sigma^*,\tau) \leq v(x;\sigma^*,\tau^*) \leq v(x;\sigma,\tau^*), \quad \forall\, \sigma, \tau \in \S. \label{Nashv}
\end{align}
If the Nash equilibrium exists, then  the  value of the game is $V(x;\sigma^*,\tau^*) = C(x) + v(x;\sigma^*,\tau^*)$, $x \in \R.$
According to \eqref{ap}, the problem is a step-up (resp.\ step-down) game when $\acheck < 0$ and $\pcheck  <0$ (resp.\ $\acheck > 0$ and $\pcheck > 0$).

\begin{remark}\label{remark_zero} If $\costb=\costs=0$, then it follows from \eqref{hx}-\eqref{fx} that $h(x)=g(x) = f(x)$ and
\[v(x; \sigma, \tau; \pcheck,  \acheck, 0, 0)  =   \E^x\left[ e^{-r(\tau\wedge \sigma)}  1_{\{X_{\tau \wedge \sigma} > 0, \, \tau \wedge \sigma < \infty \}}C(X_{\tau\wedge \sigma}; -\pcheck, -\acheck)\right].\]
In this case,  the choice of $\tau^*=\sigma^*=0$ yields the equilibrium \eqref{Nashv} with equalities, so  the default swap game is always trivially exercised at inception by either party. For similar reasons, we also rule out the trivial case with $\pcheck = 0$ or $\acheck=0$ (even with $\costs+\costb>0$).  Furthermore, we ignore the contract specifications with $\pcheck\acheck < 0$ since they mean paying more (resp. less) premium in exchange for a reduced (resp. increased) protection after exercise. Henceforth, we proceed our analysis with $\pcheck\acheck >0$ and $\costb +\costs >0$. 
\end{remark}


Next, we observe the symmetry between the step-up and step-down games.

 \begin{proposition}[symmetry]\label{prop-sym} For any  $\sigma, \tau \in \S$, we have $v(x; \sigma, \tau; \pcheck,  \acheck, \costb, \costs)  =- v(x;  \tau,\sigma; -\pcheck, -\acheck, \costs, \costb)$.
 \end{proposition}

 Applying Proposition \ref{prop-sym} to the Nash equilibrium condition \eqref{Nashv}, we deduce that if $(\sigma^*, \tau^*)$ is  the saddle point  for the step-down default swap game with $(\pcheck, \acheck, \costb, \costs)$, then the reversed pair $(\tau^*, \sigma^*)$  is the saddle point for the  step-up default swap game with $(-\pcheck, -\acheck, \costs,\costb)$. Consequently, the symmetry result implies that it is sufficient to study \emph{either}  the  step-down \emph{or} the step-up default swap game. This significantly simplifies our analysis. Henceforth,  we solve only for the step-down game.

   Also, we notice from \eqref{def_V} that if $\acheck \leq  \gamma_s$, then the seller's benefit of a reduced exposure does not exceed the fee, and therefore, should never exercise. As a result, the valuation problem is reduced to a step-down CDS studied in \cite{Leung_Yamazaki_2010}, and so we exclude it from our analysis here. With this observation and Remark \ref{remark_zero}, we will proceed with the following assumption without loss of generality:

\begin{assumption} We assume that $\acheck >  \costs\geq 0$,  $\pcheck >0$ and $\costb +\costs >0$.\end{assumption}

\subsection{Candidate Threshold Strategies} \label{subsection_continuous_smooth_fit}


In the step-down game, the protection buyer has an incentive to step-down when default  is less likely, or equivalently when $X$ is sufficiently high. On the other hand, the protection seller tends to  exercise the step-down option when default is likely to occur, or equivalently when $X$ is sufficiently small. This intuition leads us to conjecture  the following  \emph{threshold strategies}, respectively, for the buyer and the seller:
\begin{align*}
\tau_B := \inf \left\{ t \geq 0: X_t \notin (0,B) \right\}, \quad \text{and}\quad
\sigma_A := \inf \left\{ t \geq 0: X_t \notin (A,\infty) \right\}, 
\end{align*} for $B>A > 0$. Clearly, $\sigma_A, \tau_B \in \S$.  For    $B>A>0$, we denote the candidate value function
\begin{align}
v_{A,B}(x) &:= v(x; \sigma_A, \tau_B)\notag \\ &= \E^x \left[ e^{-r (\tau_B \wedge \sigma_A)} \left( h(X_{\tau_B  }) 1_{\{\tau_B < \sigma_A\}} + g(X_{  \sigma_A}) 1_{\{\tau_B > \sigma_A \}} + f(X_{\tau_B }) 1_{\{\tau_B = \sigma_A  \}}\right) 1_{\{\tau_B \wedge \sigma_A < \infty\}} \right] \notag\\
&=  \E^x \left[ e^{-r (\tau_B \wedge \sigma_A)} \left( h(X_{\tau_B}) 1_{\{\tau_B < \sigma_A\}} + g(X_{\sigma_A}) 1_{\{\tau_B > \sigma_A \}}\right) 1_{\{\tau_B \wedge \sigma_A < \infty\}} \right]\label{vbx}
\end{align}
for every $x \in \R$. The last equality follows since $\tau_B = \sigma_A$ implies that $\tau_B = \sigma_A = {\sigma_0}$, and $f(X_{\sigma_0}) =0$ a.s.

In subsequent sections, we will identify the candidate exercise thresholds $A^*$ and $B^*$
simultaneously by applying the principle of continuous and smooth fit:
\begin{align}(\text{continuous fit}) \quad v_{A,B}(B-) - h(B) &= 0 \quad \text{and}\quad v_{A,B}(A+) - g(A) = 0,\label{fit_contin}\\
(\text{smooth fit})\quad v_{A,B}'(B-) - h'(B) &= 0\quad \text{and} \quad v'_{A,B}(A+) - g'(A) = 0,\label{fit_smooth}
\end{align}
if these
limits exist.


\section{Solution Methods for the Spectrally Negative \lev Model}\label{section_spec_neg_lev}
We now define $X$ to be a \emph{spectrally negative} \lev process with  the {Laplace exponent}
\begin{align}
\phi(s) := \log \E^0 \left[ e^{s X_1} \right] =c s +\frac{1}{2}\nu^2 s^2 +\int_{(
0,\infty)}(e^{-s u}-1+s u 1_{\{0 <u<1\}})\,\Pi(\diff u), \quad  {s \in \mathbb{R}},  \label{laplace_spectrally_negative}
\end{align}
where $c \in\R$, $\nu \geq 0$ is called the Gaussian coefficient, and $\Pi$ is  a \lev measure  on $(0,\infty)$ such that\\  \mbox{$\int_{(0,\infty)} (1  \wedge u^2) \Pi( \diff u)<\infty$}. See \cite[p.212]{Kyprianou_2006}.   It admits a unique decomposition:
\begin{align} \label{x_decomposition}
X = X^c + X^d
\end{align}
where $X^c$ is the continuous martingale (Brownian motion) part and $X^d$ is the jump and drift part of $X$.  Moreover,
\begin{align}
\textrm{$X^d$ has paths of bounded variation} \Longleftrightarrow  \int_0^1 u \Pi (\diff u) < \infty. \label{bounded_var_cond}
\end{align}
If this  condition \eqref{bounded_var_cond} is satisfied,
then the Laplace exponent simplifies to \begin{align}
\phi(s) =\mu  s+  \frac{1}{2}\nu^2 s^2 + \int_{(
0,\infty)}(e^{-s u}-1)\,\Pi(\diff u), \quad s \in \mathbb{C}, \label{laplace_spectrally_negative_bounded}
\end{align}
where $
\mu := c + \int_{(
0,1)}u\, \Pi(\diff u)$. Recall that $X$ has paths of bounded variation if and only if $\nu = 0$ and \eqref{bounded_var_cond} holds.
We ignore the case when $X$ is a negative subordinator (decreasing a.s.). This means that we require $\mu$ to be strictly positive if $\nu = 0$ and \eqref{bounded_var_cond} holds. We also assume the following and also Assumption \ref{assump_scale_twice} below.
\begin{assumption} We assume that the \lev measure $\Pi$ does not have atoms.  \label{assump_atom}\end{assumption}


\subsection{Main Results} We now state our main results concerning the Nash equilibrium and its associated saddle point.  We will identify the pair of thresholds $(A^*,B^*)$ for the seller and buyer at equilibrium.  The first theorem considers the case $A^* > 0$, where the seller exercises at a level \emph{strictly above} zero.
\begin{theorem} \label{theorem_equilibrium}
Suppose  $A^* > 0$. The Nash equilibrium exists with saddle point $(\sigma_{A^*},\tau_{B^*})$ satisfying
\begin{align}
v(x; \sigma_{A^*}, \tau) \leq v_{A^*, B^*}(x) \leq v(x; \sigma, \tau_{B^*}), \quad \forall \sigma, \tau \in \S. \label{equilibrium_final}
\end{align}
\end{theorem}
Here $v_{A^*, B^*}(x) \equiv v(x;\sigma_{A^*},\tau_{B^*} )$ as in \eqref{vbx}  and can be expressed in terms of the scale function as we shall see in Subsection \ref{subsection_exp_scale}.  In particular, the case $B^*=\infty$ reflects that $\tau_{B^*} = \sigma_0$ and $v_{A^*, \infty}(x) := \lim_{B \uparrow \infty}v_{A^*,B}(x)$ is the expected value when the buyer never exercises and the seller's strategy is $\tau_{A^*}$. The value function  can be computed using \eqref{delta_by_upsilon} and \eqref{upsilon_infty} below.




The case $A^* = 0$ may occur, which  is more technical and may not yield the Nash equilibrium.  To see why, we notice that a default happens as soon as $X$ touches zero. Therefore, in the event that $X$ continuously passes (creeps) through  zero, the seller would optimally seek to exercise  at a level as close to zero as possible. Nevertheless, this timing strategy is not admissible, though it can be approximated arbitrarily closely by admissible  stopping times.

As shown in Corollary \ref{A_0_unbounded_variation} below,  the case $A^*=0$ is possible only if  the jump part $X^d$ of $X$ is of bounded variation (see \eqref{bounded_var_cond}).  This is consistent with our intuition because if $X$ jumps downward frequently, then the seller has the incentive to step down the position at a level strictly above zero.  On the other hand, when $\nu = 0$ (with no Gaussian component), the process $X$ never goes through continuously the level zero, so even with $A^*=0$ the Nash equilibrium in Theorem \ref{theorem_equilibrium} still holds. In contrast, if  $\nu > 0$, then an alternative form of ``equilibrium" is attained, namely,
\begin{align}v(x; \sigma_{0+}, \tau) \leq v_{0+,B^*}(x) \leq v(x; \sigma, \tau_{B^*}), \quad \forall \sigma, \tau \in \S, \label{Nashnot}\end{align}
where \begin{align*}
v(x; \sigma_{0+}, \tau) &:= \E^x \left[ e^{-r\tau} (h (X_\tau) - (\acheck - \gamma_s) 1_{\{X_\tau = 0\}} ) 1_{\{\tau < \infty \}} \right], \quad \tau \in \S, \\
v_{0+,B^*}(x) &:= \E^x \big[ e^{-r\tau_{B^*}} (h (X_{\tau_{B^*}}) - (\acheck - \gamma_s) 1_{\{X_{\tau_{B^*}} = 0\}} ) 1_{\{\tau_{B^*} < \infty \}} \big].
\end{align*}
Here, the functions $v(x; \sigma_{0+}, \tau)$ and $v_{0+,B^*}(x)$ correspond to the limiting case where the seller exercises  arbitrarily close to the default time $\sigma_0$.  However, since the seller cannot predict the default time, this timing strategy is not admissible and \eqref{Nashnot} is not the Nash equilibrium. In practice, given the buyer's strategy $\tau_{B^*}$, the seller's value function can be approximated with an $\varepsilon$-optimal strategy by choosing $\sigma_\delta$ for a sufficiently low exercise  level $\delta > 0$.


Let us summarize our equilibrium results for the case $A^*=0$.

\begin{theorem} \label{theorem_equilibrium_zero}
For the case $A^* = 0$,
\begin{enumerate}
\item if $\nu=0$, a Nash equilibrium exists with saddle point $({\sigma_0},\tau_{B^*})$ and \eqref{equilibrium_final} holds;
\item if $\nu > 0$, then  the alternative equilibrium \eqref{Nashnot} holds.
\end{enumerate}
\end{theorem}

In the remainder of this section, we take the following steps to prove the existence of $(A^*,B^*)$ and Theorems \ref{theorem_equilibrium}-\ref{theorem_equilibrium_zero}:
\begin{enumerate}
\item In Section \ref{subsection_exp_scale}, we express the candidate value function $v_{A,B}$ in terms of the \lev scale function.
\item In Section \ref{subsection_cont_smooth_fit}, we establish the sufficient conditions for continuous and smooth fit.
\item In Section \ref{subsection_existence}, we show the existence of the candidate optimal thresholds $A^*$ and $B^*$ (Theorem \ref{lemma_classification_b}).
\item In Section \ref{subsection_verification}, we verify the optimality of the candidate optimal exercise strategies.
\end{enumerate}
Furthermore, in Section  \ref{subsection_existence}  we provide an efficient algorithm to compute the pair $(A^*,B^*)$ and $v_{A^*,B^*}(x)$.  Finally, with Theorems \ref{theorem_equilibrium}-\ref{theorem_equilibrium_zero}, the value of the step-down game is recovered by $V(x) = C(x) + v(x)$ by Proposition \ref{prop-V} and that of the step-up game is recovered by $V(x) = C(x) - v(x)$  by Proposition \ref{prop-sym}.

\begin{remark}\label{rem-RNeu}
For the fair valuation of the default swap game, one may specify $\p$ as the risk-neutral pricing measure. The risk-neutrality condition would require that $\phi(1)=r$ so that the discounted asset value is a $(\p, \mathbb{F})$-martingale. This condition is not needed for our solution approach and equilibrium results.
\end{remark}

\subsection{Expressing $v_{A,B}$ using the scale function.} \label{subsection_exp_scale}
In this subsection, we shall summarize the scale function associated with the process $X$, and then  apply this to compute the candidate value function $v_{A,B}(x)$ defined in \eqref{vbx}.  For any spectrally negative \lev process, there exists a function $W^{(r)}: \R \mapsto \R$, which is zero on $(-\infty,0)$ and continuous and strictly increasing on $[0,\infty)$. It is characterized by the Laplace transform:
\begin{align*}
\int_0^\infty e^{-s x} W^{(r)}(x) \diff x = \frac 1
{\phi(s)-r}, \qquad s > \lapinv,
\end{align*}
where $\Phi$ is the right inverse of $\phi$, defined by
\begin{equation}
\lapinv :=\sup\{\lambda \geq 0: \phi(\lambda)=r\}. \notag
\end{equation}
The function $W^{(r)}$ is often called  the  \emph{(r-)scale function} in the literature (see e.g. \cite{Kyprianou_2006}).

With $\lapinv$ and $W^{(r)}$, we  can define the function $W_{\lapinv} = \{ W_{\lapinv} (x); x \in \R \}$ by
\begin{align}
W_{\lapinv} (x)  = e^{ - \lapinv x}W^{(r)} (x) , \quad x \in \R.  \label{w_phi}
\end{align}
As is well known (see  \cite[Chapter 8]{Kyprianou_2006}), the function $W_{\lapinv}(x)$ is increasing, and satisfies \begin{align}
W_{\lapinv}(x) &\uparrow \frac 1 {\phi'(\lapinv)} \quad \textrm{as } \; x \uparrow\infty. \label{scale_function_asymptotic}
\end{align}
From Lemmas 4.3-4.4 of \cite{Kyprianou_Surya_2007}, we also summarize the behavior of $W^{(r)}$ in the neighborhood of zero:
\begin{align}
W^{(r)} (0) = \left\{ \begin{array}{ll} 0, & \textrm{unbounded variation} \\ \frac 1 {\mu}, & \textrm{bounded variation} \end{array} \right\} \quad \textrm{and} \quad
W^{(r)'} (0+) = \left\{ \begin{array}{ll}  \frac 2 {\nu^2}, & \nu > 0 \\   \infty, & \nu = 0 \; \textrm{and} \; \Pi(0,\infty) = \infty \\ \frac {r + \Pi(0,\infty)} {\mu^2}, & \textrm{compound Poisson} \end{array} \right\}. \label{lemma_zero}
\end{align}

To facilitate calculations, we   define the function
\begin{align*}
Z^{(r)}(x) := 1 + r \int_0^x W^{(r)}(y) \diff y, \quad x \in \R
\end{align*}
which satisfies that
\begin{align}
\frac {Z^{(r)}(x)}  {W^{(r)}(x)} \xrightarrow{x \uparrow \infty} \frac r {\lapinv}; \label{convergence_z_by_w}
\end{align}
see \cite{Kyprianou_2006} Exercise 8.5. By Theorem 8.5 of   \cite{Kyprianou_2006}, the  Laplace transform of ${\sigma_0}$ in (\ref{zeta}) can be expressed as
\begin{align}
\zeta(x) = Z^{(r)}(x) - \frac r {\lapinv} W^{(r)}(x), \quad x > 0. \label{zeta_by_scale_function}
\end{align}
Regarding the smoothness of the scale function, 
Assumption \ref{assump_atom} guarantees that $W^{(r)}(x)$ is differentiable on $(0,\infty)$ (see, e.g., \cite{Chan_2009}). By \eqref{zeta_by_scale_function}, Laplace transform function $\zeta$ is also differentiable on $(0,\infty)$, and so are the functions $h, g, f$ in \eqref{hx}-\eqref{fx}.  
In this paper, we need the twice differentiability for the case of unbounded variation.
\begin{assumption}  \label{assump_scale_twice}For the case $X$ is of unbounded variation, we assume that $W^{(r)}$ is twice differentiable on $(0, \infty)$. 
\end{assumption}
This assumption is automatically satisfied if $\nu > 0$  as in \cite{Chan_2009}, and the same property holds for $\zeta, h, g$, and $f$.  While this is not guaranteed for the unbounded variation case with $\nu = 0$, it is an assumption commonly needed when the verification of optimality requires the infinitesimal generator.  

Moreover, as in (8.18) of \cite{Kyprianou_2006},
\begin{align}
 \frac {W^{(r)'}(y)} {W^{(r)}(y)} \leq \frac {W^{(r)'}(x)} {W^{(r)}(x)} \quad \textrm{and} \quad \frac {W^{'}_{\Phi(r)}(y)} {W_{\Phi(r)}(y)} \leq \frac {W^{'}_{\Phi(r)}(x)} {W_{\Phi(r)}(x)},  \quad y > x > 0, \label{assumeW}
\end{align}
and, using \eqref{scale_function_asymptotic}, we deduce that
\begin{align}
\frac  {W^{(r)'}(x)} {W^{(r)}(x)} = \frac  {\Phi(r) e^{\Phi(r)x}W_{\Phi(r)}(x) + e^{\Phi(r)x}W_{\Phi(r)}'(x)} {e^{\Phi(r)x}W_{\Phi(r)}(x)} = \frac  {\Phi(r) W_{\Phi(r)}(x) + W_{\Phi(r)}'(x)} {W_{\Phi(r)}(x)} \xrightarrow{x \uparrow \infty}  {\Phi(r)}. \label{asymptotics_ratio_with_derivative}
\end{align}

In applying the scale function to compute $v_{A,B}(x)$, we first consider the case $0 < A < B < \infty$ and then extend to the cases $A \downarrow 0$ and $B \uparrow \infty$, namely,
\begin{align}
v_{A, \infty}(x) := \lim_{B \uparrow \infty}v_{A,B}(x) \quad \textrm{and} \quad v_{0+, B}(x) := \lim_{A \downarrow 0}v_{A,B}(x). \label{v_limit}
\end{align} For $0<A<x < B<\infty$, define
\begin{align}
\begin{split}
\Upsilon(x;A,B) &:= \Big( \frac \pcheck r - \costb \Big) \E^x\left[  e^{-r (\sigma_A \wedge \tau_B)} 1_{\{\tau_B < \sigma_A \}} \right]  +\Big( \frac \pcheck r + \costs \Big) \E^x\left[  e^{-r (\sigma_A \wedge \tau_B)} 1_{\{\tau_B > \sigma_A \, \textrm{or} \;  \sigma_A \wedge \tau_B = {\sigma_0} \}} \right] \\ &\; +(\acheck-\costs) \E^x\left[  e^{-r (\sigma_A \wedge \tau_B)}1_{\{\sigma_A \wedge \tau_B = {\sigma_0} \}} \right]. \end{split} \label{definition_upsilon}
\end{align}
We observe that $v_{A,B}(x) - h(x)$ and $v_{A,B}(x) - g(x)$ are similar and they  possess the common term $\Upsilon(x;A,B)$.
\begin{lemma} \label{lemma_delta_b} For $0<A<x < B<\infty$, 
\begin{align} \label{delta_by_upsilon}
\begin{split}
v_{A,B}(x) - h(x)  &=\Upsilon(x;A,B)  - \Big( \frac \pcheck r - \costb \Big), \\
v_{A,B}(x) - g(x)  &=\Upsilon(x;A,B)  - \Big( \frac \pcheck r + \costs \Big),
\end{split}
\end{align}
and
\begin{align}
\Upsilon(x;A,B) &=  W^{(r)}(x-A) \frac {\Psi(A,B)} {W^{(r)}(B-A)} + \Big( \frac \pcheck r + \costs \Big)  Z^{(r)}(x-A) - \left( \acheck - \costs \right) \kappa(x;A),  \label{upsilon_in_terms_of_scale_function}
\end{align}
 where
\begin{align}
\Psi(A,B) &:= \Big( \frac \pcheck r - \costb \Big) - \Big( \frac \pcheck r +\costs \Big) Z^{(r)} (B-A) + \left( \acheck - \costs \right) \kappa(B;A), \quad 0 < A < B < \infty, \label{def_Psi} \\
\begin{split}
\kappa(x;A) &:= \int_A^{\infty} \Pi(\diff u) \int_0^{u\wedge x-A} \diff z W^{(r)}(x-z-A) \\&\; = \frac 1 r \int_A^{\infty} \Pi(\diff u)  \left[ Z^{(r)}(x-A) - Z^{(r)}(x-u) \right], \quad x > A > 0.  \end{split}\label{def_kappa}
\end{align}
\end{lemma}
The function $\Psi(A,B)$ as in \eqref{def_Psi}  will play a crucial role in the continuous and smooth fit as we discuss in Subsection \ref{subsection_cont_smooth_fit} below and also in the proof of the existence of a pair  $(A^*,B^*)$ as in Subsection \ref{subsection_existence}.

Now we extend our definition of $v_{A,B}$ for $A=0+$ and $B = \infty$ as in \eqref{v_limit}, and then derive the strategies that attain them.
As we shall see in Corollary \ref{A_0_unbounded_variation} below, our candidate threshold level for the seller $A^*$ is always strictly positive if $X^d$ is of unbounded variation whether or not there is a Gaussian component.  For this reason, we consider the limit as $A \downarrow 0$ only when \eqref{bounded_var_cond} is satisfied.

In view of \eqref{delta_by_upsilon}, the limits in \eqref{v_limit} can be obtained by extending $\Upsilon(x;A,B)$ with $A \downarrow 0$ and $B \uparrow \infty$; namely we take limits in \eqref{upsilon_in_terms_of_scale_function}.   Here $\Psi$ as in \eqref{def_Psi} explodes as $B \uparrow \infty$ and hence we define an extended version of ${\Psi (A,B)}/ {W^{(r)}(B-A)}$ by, for any $0 \leq A < B \leq \infty$ (with the assumption $\int_0^1 u \Pi (\diff u) < \infty$ for $A = 0$),
\begin{align} \label{gamma_a_definition}
  \widehat{\Psi}(A,B)  &:= \left\{  \begin{array}{ll}  \frac 1 {W^{(r)}(B-A)} \left[\left( \frac \pcheck r - \costb \right) - \left( \frac \pcheck r +\costs \right) Z^{(r)} (B-A) + \left( \acheck - \costs \right) \kappa(B;A) \right], & B < \infty, \\ \frac 1 {\Phi(r)} \left( - (\pcheck + r \costs) + (\acheck - \costs) \rho(A) \right),  &  B = \infty, \end{array} \right.
\end{align}
where\begin{align*}
\rho(A) := \int_A^\infty \Pi (\diff u)  \left( 1 -  e^{-\lapinv (u-A)} \right) = \int_0^\infty \Pi (\diff u + A)  \left( 1 -  e^{-\lapinv u} \right), \quad A \geq 0
\end{align*}
and
\begin{align}
\kappa(x;0) := \int_0^{\infty} \Pi(\diff u) \int_0^{u\wedge x} \diff z W^{(r)}(x-z) = \frac 1 r \int_0^{\infty} \Pi(\diff u)  \left[ Z^{(r)}(x) - Z^{(r)}(x-u) \right], \quad x > 0. \label{kappa_zero}
\end{align}
Here, $\rho(0) = \int_0^\infty \Pi (\diff u)  \left( 1 -  e^{-\lapinv u} \right)$
is finite if and only if \eqref{bounded_var_cond} holds.
Clearly, $\widehat{\Psi} (A,B) = \frac {\Psi (A,B)} {W^{(r)}(B-A)}$ when $0 < A < B < \infty$.
We shall confirm the convergence results and other auxiliary results below.

%

\begin{lemma}\label{convergence_kappa} For any fixed $x > 0$,
\begin{enumerate}
\item $\kappa(x;A)$ is monotonically decreasing in $A$ on $(0,x)$,
\item  if $\int_0^1 u \Pi (\diff u) < \infty$, then $\kappa(x;0) = \lim_{A \downarrow 0} \kappa(x;A) < \infty$,
\item for every $A > 0$ (extended to  $A\geq 0$ if $\int_0^1 u \Pi (\diff u) < \infty$),
$\frac {\kappa(x; A)} {W^{(r)}(x-A)} \xrightarrow{x \uparrow \infty} \frac {\rho(A)} {\Phi(r)}$.
\end{enumerate}
\end{lemma}


\begin{lemma}\label{remark_asymptotics_Psi}
\begin{enumerate}
\item We have $\lim_{B \uparrow \infty}\widehat{\Psi}(A,B) = \widehat{\Psi}(A,\infty)$ for every $A > 0$ (extended to  $A\geq 0$ if $\int_0^1 u \Pi (\diff u) < \infty$).
\item  When $\int_0^1 u \Pi (\diff u) < \infty$, for every $0 < B < \infty$ and $0 < B \leq \infty$, respectively,
\begin{align}
\lim_{A \downarrow 0}\Psi(A,B) =\Big( \frac \pcheck r - \costb \Big) - \Big( \frac
\pcheck r +\costs \Big) Z^{(r)} (B) + \left( \acheck - \costs
\right) \kappa(B;0) =: \Psi(0, B), \label{Psi_big_zero}
\end{align}
and $\widehat{\Psi}(0,B)  =  \lim_{A \downarrow 0}\widehat{\Psi}(A,B)$.
\item For every $A > 0$ (extended to  $A\geq 0$ if $\int_0^1 u \Pi (\diff u) < \infty$), $\Psi(A,A+) < 0$.
\end{enumerate}
\end{lemma}

Using the above, for $0 < A < x$, we obtain the limit
\begin{align}
\Upsilon(x;A,\infty) := \lim_{B \uparrow \infty}\Upsilon(x;A,B) = W^{(r)}(x-A) \widehat{\Psi}(A,\infty) + \Big( \frac \pcheck r + \gamma_s \Big) Z^{(r)}(x-A) - \left( \acheck - \gamma_s \right) \kappa(x;A),   \label{upsilon_infty}
\end{align}
and, for $0 < x < B \leq \infty$,
\begin{align}
\Upsilon(x; 0+,B)  := \lim_{A \downarrow 0}\Upsilon(x; A,B)  =  W^{(r)}(x) \widehat{\Psi}(0,B) + \Big( \frac \pcheck r + \costs \Big)  Z^{(r)}(x) - \left( \acheck - \costs \right) \kappa(x;0). \label{upsilon_zero_via_scale}
\end{align}
In summary, we have expressed $v_{A,B}$ including its limits in \eqref{v_limit} in terms of the scale function.

\begin{remark} \label{remark_smoothness}We note that $v_{A,B}(x)$ is $C^1(A,B)$ and in particular $C^2(A,B)$ when $X$ is of unbounded variation.  Indeed, $\kappa(x;A)$ is $C^1(A,B)$ and in particular $C^2(A,B)$ when $X$ is of unbounded variation.  See also the discussion immediately before and after Assumption \ref{assump_scale_twice} for the same smoothness property on $(0,\infty) \backslash [A,B]$.
\end{remark}

We now construct the strategies that achieve $v_{A, \infty}(x)$ and $v_{0+, B}(x)$.  As the following remark shows, the interpretation of the former is fairly intuitive and it is attained when the buyer never exercises and his strategy is $\sigma_0$.
\begin{remark} \label{remark_asymptotics}
By \eqref{zeta_by_scale_function} and Lemma 3.4 of \cite{Leung_Yamazaki_2010}, respectively, we have, for any $A > 0$, $\E^x\left[  e^{-r \sigma_A} \right] = Z^{(r)} (x-A) - \frac r {\Phi(r)} W^{(r)}(x-A)$ and $ \E^x\left[  e^{-r \sigma_A}1_{\{\sigma_A  = {\sigma_0}  < \infty \}} \right] = {W^{(r)}(x-A)}   \frac {\rho(A)} {\Phi(r)} - \kappa(x;A)$
and hence it can be confirmed from \eqref{upsilon_infty} that
\begin{align*}
\Upsilon(x;A,\infty)  = \Big( \frac \pcheck r + \costs \Big) \E^x\left[  e^{-r \sigma_A} \right]+  (\acheck-\costs) \E^x\left[  e^{-r \sigma_A}1_{\{\sigma_A  = {\sigma_0} < \infty\}} \right],
\end{align*}
which corresponds to the value when the buyer's strategy is $\sigma_0$ and the seller's strategy is $\sigma_A$.
\end{remark}

On the other hand, $v_{0+, B}(x)$ is slightly more difficult to understand.
Suppose we substitute $A=0$ directly into \eqref{definition_upsilon} (or the seller never exercises and her strategy is $\sigma_0$), we obtain
\begin{align*}
\Upsilon(x;0,B) := \Big( \frac \pcheck r - \gamma_b \Big) \E^x \left[ e^{-r \tau_B} 1_{\{\tau_B < {\sigma_0}, \, \tau_B < \infty \}}\right] + \Big( \frac \pcheck r + \acheck \Big) \E^x \left[ e^{-r \tau_B} 1_{\{\tau_B = {\sigma_0} < \infty\}}\right], \quad 0 < B \leq \infty.
\end{align*}
As shown  in Remark \ref{lemma_creeping} below,  $\Upsilon(x;0,B)$ matches $\Upsilon(x;0+,B)$ if and only if there is not a Gaussian component. Upon the existence of Gaussian component, there is a positive probability of continuously down-crossing (creeping) zero, and the seller tends to exercise \emph{immediately before} it reaches zero rather than not exercising at all.

\begin{remark} \label{lemma_creeping} The right-hand limit  $\Upsilon(x;0+,B) := \lim_{A \downarrow 0}\Upsilon(x;A,B)$ is given by
\begin{align}
\label{upsilon_zero}\Upsilon(x;0+,B) = \Upsilon(x;0,B) - \left( \acheck - \gamma_s \right) \E^x \left[ e^{-r \tau_B}  1_{\{X_{\tau_B} = 0, \, \tau_B < \infty\}} \right], \quad 0 < x < B \leq \infty.
\end{align}
Therefore, $\Upsilon(x;A,B) \xrightarrow{A \downarrow 0} \Upsilon(x;0,B)$ if and only if the Gaussian coefficient $\nu = 0$.
\end{remark}


Upon the existence of a Gaussian component, $\Upsilon(x;0,B) > \Upsilon(x;0+,B)$, but there does not exist a seller's strategy that attains $v_{0+,B}$.  However, for any $\varepsilon > 0$, the $\varepsilon$-optimal strategy (when the buyer's strategy is $\tau_B$) can be attained by choosing a sufficiently small level.  Without a Gaussian component, $\Upsilon(x;0,B) = \Upsilon(x;0+,B)$ and the seller may choose ${\sigma_0}$.

\subsection{Continuous and Smooth Fit} \label{subsection_cont_smooth_fit}
We shall now find the candidate thresholds $A^*$ and $B^*$ by continuous
and smooth fit. As we will show below,
 the continuous and smooth fit conditions \eqref{fit_contin}-\eqref{fit_smooth} will yield the equivalent
conditions $\Psi(A^*,B^*) = \psi(A^*,B^*)=0$ where
\begin{align*}
\psi(A,B) := \frac \partial {\partial B}\Psi(A,B) =   - W^{(r)} (B-A) \left( \pcheck +\costs r \right)  + \left( \acheck - \costs \right) \int_A^\infty \Pi(\diff u)  \left( W^{(r)}(B-A) - W^{(r)}(B-u) \right),
\end{align*}
for all $0 < A < B < \infty$.  Here the second equality holds because for every $x > A > 0$
\begin{align}
\begin{split}
Z^{(r)'}(x-A) = r W^{(r)}(x-A) \quad \textrm{and} \quad \kappa'(x;A) = \int_A^\infty \Pi (\diff u) \left( W^{(r)}(x-A) - W^{(r)}(x-u) \right), \end{split}\label{derivative_z_kappa}
\end{align}
where the latter holds because $Z^{(r)'}(x) = r W^{(r)}(x)$ on $\R \backslash \{0\}$ and $Z^{(r)}$ is continuous on $\R$.


As in the case of $\Psi(A,\cdot)$, it can be seen that $\psi(A,\cdot)$ also tends to explode as $B \uparrow \infty$ with $A$ fixed.  For this reason, we also define the extended version of ${\psi (A,B)}/{W^{(r)}(B-A)}$ by, for any $0 \leq A < B \leq \infty$ (with the assumption $\int_0^1 u \Pi (\diff u) < \infty$ for $A = 0$),
\begin{align}
\label{gamma_b_definition}
\widehat{\psi}(A,B) &:= \left\{ \begin{array}{ll} -\left( \pcheck + \costs r \right) +  \left( \acheck - \costs \right)\int_{A}^{\infty} \Pi (\diff u) \left( 1 - \frac {W^{(r)}(B-u)} {W^{(r)}(B-A)} \right), & B < \infty, \\ -\left( \pcheck + r \costs \right)  +  (\acheck - \costs)  \rho(A), & B = \infty. \end{array} \right.
\end{align}
The convergence results as $A \downarrow 0$ and $B \uparrow \infty$ as well as some monotonicity properties are discussed below.

\begin{lemma}\label{lemma_gamma_B}
\begin{enumerate}
\item  For fixed $0 < B \leq \infty$, $\widehat{\psi}(A,B)$ is decreasing  in $A$ on $(0,B)$, and in particular when $\int_0^1 u \Pi (\diff u) < \infty$, $\widehat{\psi}(0,B) = \lim_{A \downarrow 0}\widehat{\psi}(A,B)$.
\item For fixed $A > 0$ (extended to  $A\geq 0$ if $\int_0^1 u \Pi (\diff u) < \infty$), $\widehat{\psi}(A,B)$ is decreasing in $B$ on $(A,\infty)$ and $\widehat{\psi}(A,B) \downarrow \widehat{\psi}(A,\infty)$ as $B \uparrow \infty$.
\item The relationship $\psi(0,B) = \partial \Psi(0,B)/ {\partial B}$ also holds for any $0 < B < \infty$ given $\int_0^1 u \Pi (\diff u) < \infty$ where $\Psi(0,B)$ is defined as in \eqref{Psi_big_zero} and 
\begin{align*}
 \psi(0, B) &:=W^{(r)}(B)\Big[-(\tilde{p}+\gamma_sr)+(\tilde{\alpha}-\gamma_s)\int_0^\infty
\Pi(\diff u)\Big(1-\frac{W^{(r)}(B-u)}{W^{(r)}(B)}\Big)\Big].
\end{align*}
\end{enumerate}
\end{lemma}


\begin{figure}[htbp]
\begin{center}
\begin{minipage}{1.0\textwidth}
\centering
\begin{tabular}{cc}
\includegraphics[scale=0.6]{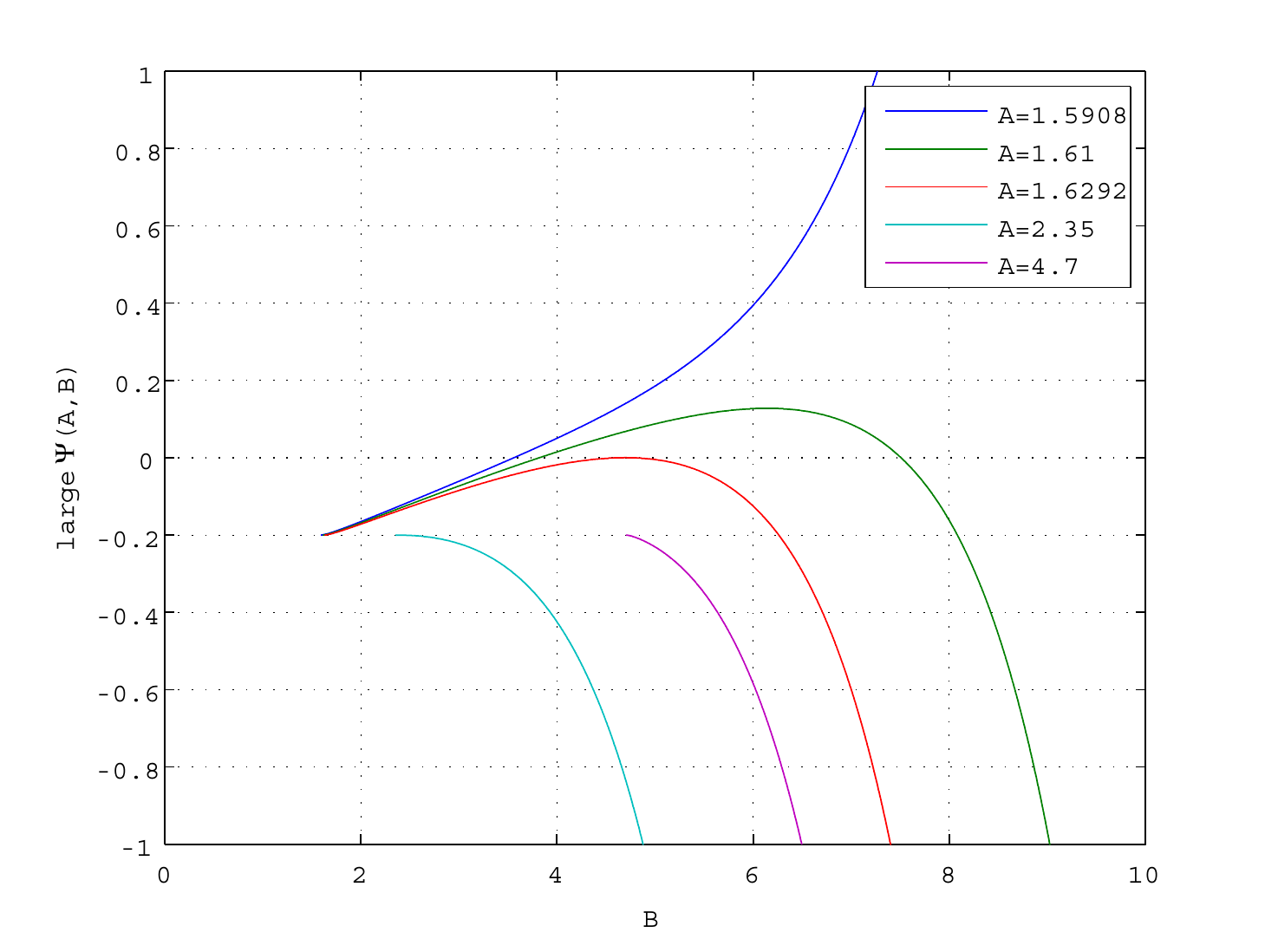}  & \includegraphics[scale=0.6]{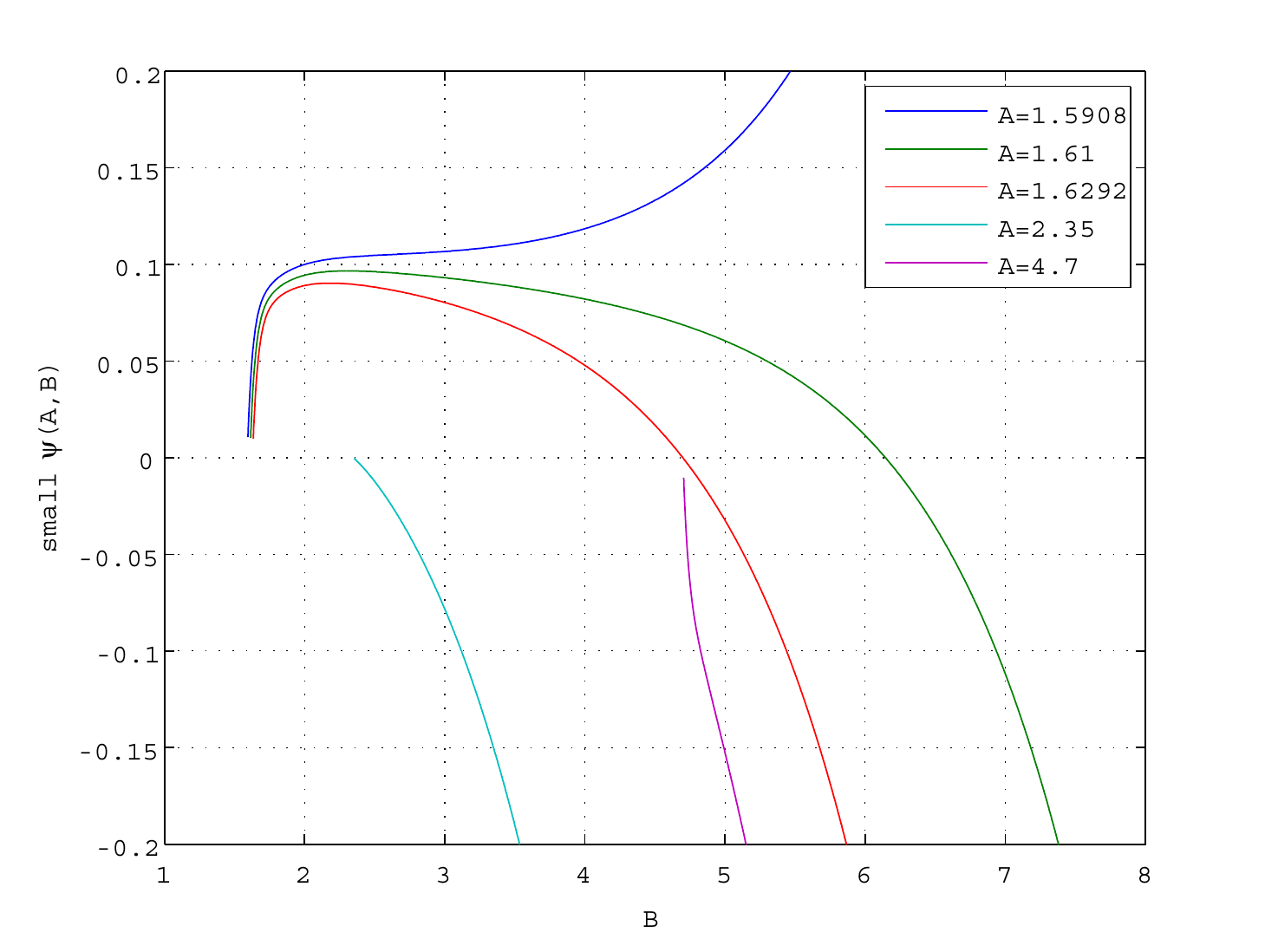} \\
$\Psi(A,B)$  & $\psi(A,B)$ \\
\includegraphics[scale=0.6]{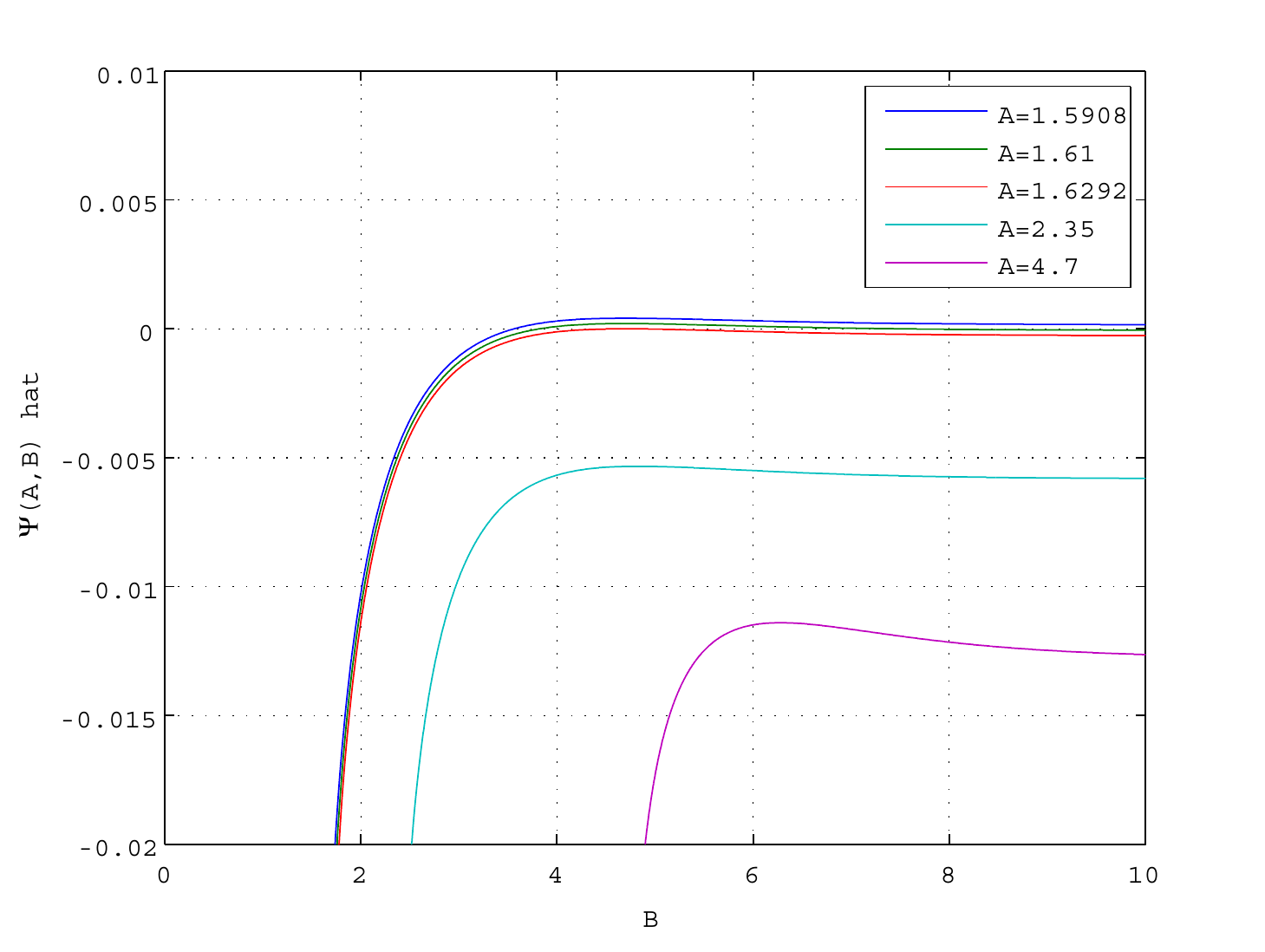}  & \includegraphics[scale=0.6]{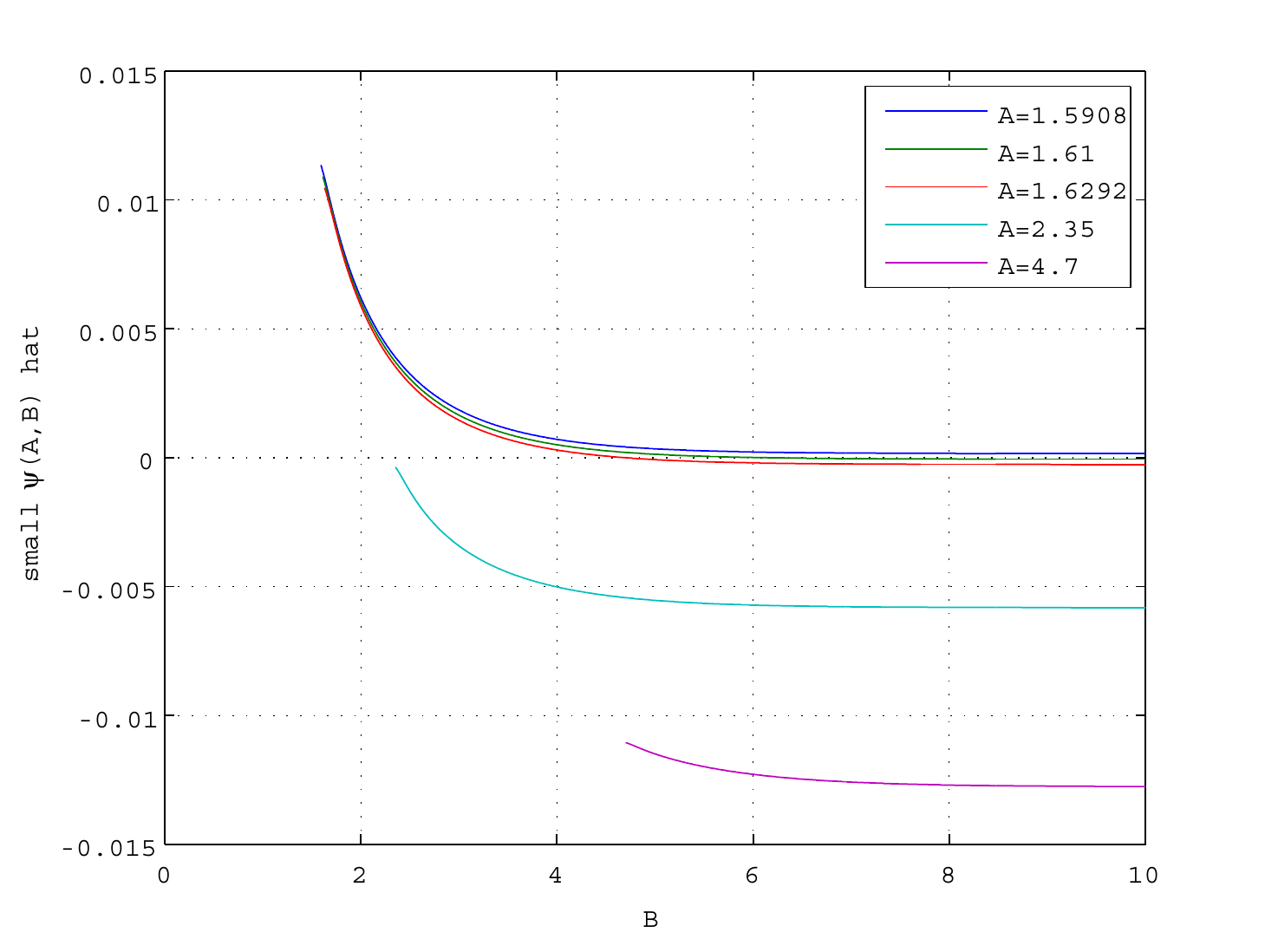} \\
$\widehat{\Psi}(A,B)$  & $\widehat{\psi}(A,B)$
\end{tabular}
\end{minipage}
\caption{\small Illustration of $\Psi(A,B)$,  $\widehat{\Psi}(A,B)$, $\psi(A,B)$, and $\widehat{\psi}(A,B)$ as functions of $B$. } \label{figure_psi}
\end{center}
\end{figure}

Figure \ref{figure_psi} gives numerical plots of $\Psi(A,\cdot)$, $\widehat{\Psi}(A,\cdot)$, $\psi(A,\cdot)$  and $\widehat{\psi}(A,\cdot)$ for various values of $A > 0$.   Lemma \ref{lemma_gamma_B}-(1,2) and the fact that $\widehat{\psi}(A,B) \geq 0 \Longleftrightarrow \psi(A,B) \geq 0$ imply that, given a fixed $A$, there are three possible behaviors for $\Psi$:
\begin{enumerate}
\item[(a)] For small $A$, $\Psi(A,B)$ is monotonically increasing in $B$.
\item[(b)] For large $A$, $\Psi(A,B)$ is monotonically decreasing in $B$.
\item[(c)] Otherwise $\Psi(A,B)$ first increases  and then decreases in $B$.
\end{enumerate}
{The behavior of $\Psi$ has implications for the existence and uniqueness of $A^*$ and $B^*$, as shown in Theorem \ref{lemma_classification_b} and Lemma \ref{lemma_b_bar} below.} Besides, it can be confirmed that $\widehat{\Psi}(A,\cdot)$ and $\widehat{\psi}(A,\cdot)$ converge as $B \uparrow \infty$ as in Lemmas \ref{remark_asymptotics_Psi}-(1) and \ref{lemma_gamma_B}-(2).  We shall see that the continuous/smooth fit conditions \eqref{fit_contin}-\eqref{fit_smooth} require (except for the case $A^*=0$ or $B^* = \infty$) that $\widehat{\Psi}(A^*,B^*)=\widehat{\psi}(A^*,B^*)=0$, or equivalently $\Psi(A^*,B^*)=\psi(A^*,B^*)=0$. This is illustrated by the line corresponding to $A = 1.6292$ in Figure  \ref{figure_psi}.

We begin with establishing the continuous fit condition.

\underline{Continuous fit at $B$:}   Continuous fit at $B$ is satisfied automatically for all cases since $v_{A,B}(B-)-h(B)$ exists and
\begin{align}
v_{A,B}(B-)-h(B) = \Upsilon(B-; A, B) - \Big( \frac \pcheck r - \costb \Big) =0, \quad 0 < A < B < \infty, \label{eq_continuous_fit_B}
\end{align}
which also holds when $A=0+$ and $v_{0+,B}(B-)-h(B)  = 0$ given $\int_0^1 u \Pi (\diff u) < \infty$.
This is also clear from the fact that a spectrally negative \lev process always creeps upward and hence $B$ is regular for $(B,\infty)$ for any arbitrary level $B > 0$ (see \cite[p.212]{Kyprianou_2006}).

\underline{Continuous fit at $A$:} We examine the limit of $v_{A,B}(x)-g(x)$ as $x\downarrow A$, namely,
\begin{align}
v_{A,B}(A+)-g(A)  =W^{(r)}(0) \widehat{\Psi}(A,B), \quad 0 < A < B \leq \infty. \label{eq_continuous_fit_A}
\end{align}
In view of \eqref{lemma_zero}, continuous fit at $A$ holds automatically for the unbounded variation case.  For the bounded variation case, the continuous fit condition is equivalent to
\begin{align}
\widehat{\Psi}(A,B)  = 0. \label{fit_A}
\end{align}


We now pursue the  smooth fit condition.   Substituting \eqref{derivative_z_kappa} into the derivative of \eqref{upsilon_in_terms_of_scale_function}, we obtain
\begin{align}
v_{A,B}'(x)-h'(x) = v_{A,B}'(x)-g'(x) = \Upsilon'(x;A,B) = W^{(r)'}(x-A) \widehat{\Psi}(A,B)  - \psi(A,x), \label{deltaprimes}
\end{align}
for every $0 < A < x < B \leq \infty$ (extended to $A=0+$ when $\int_0^1 u \Pi(\diff u) < \infty$). \\
\underline{Smooth fit at $B$:} With  \eqref{deltaprimes},  the smooth fit condition $v_{A,B}'(B-)-h'(B) = 0$ at $B < \infty$ amounts to
\begin{align*}
\frac \partial {\partial B} \widehat{\Psi}(A,B)  = 0
\end{align*}
because
\begin{align*}
W^{(r)'}(B-A) \widehat{\Psi}(A,B)  - \psi(A,B) = - W^{(r)}(B-A) \Big(\widehat{\psi}(A,B) - \frac {W^{(r)'}(B-A)} {W^{(r)}(B-A)} \widehat{\Psi}(A,B) \Big)
\end{align*}
and
\begin{align}
\frac \partial {\partial B} \widehat{\Psi}(A,B) = \widehat{\psi}(A,B) - \frac {W^{(r)'}(B-A)} {W^{(r)}(B-A)} \widehat{\Psi}(A,B). \label{Psi_hat_derivative}
\end{align}

 For the case $A=0+$ and $\int_0^1 u \Pi (\diff u) < \infty$,  the smooth fit condition $v_{0+,B}'(B-)-h'(B) =0$ requires $\frac \partial {\partial B} \widehat{\Psi}(0,B)=0$, which is well-defined by Lemmas \ref{remark_asymptotics_Psi}-(2) and \ref{lemma_gamma_B}-(1) and \eqref{Psi_hat_derivative}.

%
%
%
\underline{Smooth fit at $A$:}  Assuming that it has paths of unbounded variation ($W^{(r)}(0)=0$), then we obtain
\begin{align*}
v_{A,B}'(A+)-g'(A)  &= W^{(r)'}(0+) \widehat{\Psi}(A,B), \quad 0 < A < B \leq \infty.
\end{align*}
Therefore, \eqref{fit_A} is also a sufficient condition for smooth fit at $A$ for the unbounded variation case.


We conclude that
\begin{enumerate}
\item if $\widehat{\Psi}(A,B)=0$, then continuous fit  at $A$ holds for the bounded variation case and both continuous and smooth fit  at $A$ holds for the unbounded variation case;
\item if $\frac \partial {\partial B} \widehat{\Psi}(A,B) =0$, then both continuous and smooth fit conditions at $B$  hold for all cases.
\end{enumerate}
If both $\widehat{\Psi}(A,B)=0$ and $\frac \partial {\partial B} \widehat{\Psi}(A,B)  = 0$ are satisfied, then $\widehat{\psi}(A,B)=0$ automatically follows by \eqref{Psi_hat_derivative}.



\subsection{Existence and Identification of $(A^*,B^*)$} \label{subsection_existence}In the previous subsection, we have derived the defining equations for  the candidate pair $(A^*,B^*)$. Nevertheless, the computation of $(A^*,B^*)$ is  non-trivial and depends on the behaviors of functions $\Psi(A,B)$ and $\psi(A,B)$. In this subsection, we prove the existence of  $(A^*,B^*)$ and provide a procedure to calculate their values.

Recall from Lemma \ref{lemma_gamma_B}-(1) that $\widehat{\psi}(A,\infty)$ is decreasing in $A$ and observe that $\widehat{\psi}(A,A+) := \lim_{x \downarrow A}\widehat{\psi}(A,x) = -(\pcheck + r \gamma_s) + (\acheck - \gamma_s) \Pi(A,\infty)$ is also decreasing in $A$.  Hence, let $\underline{A}$ and $\overline{A}$ be the unique values such that
\begin{align}
\widehat{\psi}(\underline{A},\infty) &\equiv -\left( \pcheck + r \costs \right)  + (\acheck - \costs) \rho(\underline{A}) = 0, \label{psi_A_under_over}\\
\widehat{\psi}(\overline{A},\overline{A}+) &\equiv -\left( \pcheck + r \costs \right)  + (\acheck - \costs) \Pi (\overline{A},\infty)  = 0,\label{psi_A_under_over2}
 \end{align}
upon existence; we set the former zero if $\widehat{\psi}(A,\infty) < 0$ for all $A \geq 0$ and also set the latter zero if $\widehat{\psi}(A,A+) < 0$ for any $A \geq 0$.  Since $\rho(A) \downarrow 0$ and $\Pi(A,\infty) \downarrow 0$ as $A \uparrow \infty$, $\overline{A}$ and $\underline{A}$ are finite. In addition,  $\rho(A) < \Pi(A,\infty)$ implies that $\overline{A} \geq \underline{A}$.

Define for every $\underline{A} \leq A \leq \overline{A}$,
\begin{align} \label{def_b}
\begin{split}
&\underline{b}(A) := \inf \{ B > A: \widehat{\Psi}(A,B) \geq 0 \} \equiv \inf \left\{ B > A: \Psi(A,B) \geq 0 \right\}, \\  &\overline{b}(A) := \inf \{ B > A: \widehat{\psi}(A,B) \leq 0 \} \equiv \inf \left\{ B > A: \psi(A,B) \leq 0 \right\}, \\ &b(A) := \inf  \{ B > A:  \widehat{\Psi}(A,B) - \widehat{\psi}(A,B) \frac {W^{(r)}(B-A)} {W^{(r)'}(B-A)} \geq 0 \},
\end{split}
\end{align}
where we assume $\inf \emptyset = \infty$.  For $b(A)$ above,  we recall from  \eqref{Psi_hat_derivative} that
\begin{align}
\widehat{\Psi}(A,B) - \widehat{\psi}(A,B) \frac {W^{(r)}(B-A)} {W^{(r)'}(B-A)} = 0 \Longleftrightarrow  \frac \partial {\partial B} \widehat{\Psi}(A,B)  = 0. \label{equivalence_psi_hat}
\end{align}
  Also, using Lemmas \ref{remark_asymptotics_Psi}-(1) and \ref{lemma_gamma_B}-(2) and that  $\Phi(r) \widehat{\Psi}(A,\infty) = \widehat{\psi}(A,\infty)$ (see \eqref{asymptotics_ratio_with_derivative}, \eqref{gamma_a_definition}, and \eqref{gamma_b_definition}), we obtain the limit
\begin{align}
\lim_{B \uparrow \infty}\Big( \widehat{\Psi}(A,B) - \widehat{\psi}(A,B) \frac {W^{(r)}(B-A)} {W^{(r)'}(B-A)}\Big) = 0. \label{eq_behavior_gamma}
\end{align}


Next, we show that there always exists a pair $(A^*,B^*)$ belonging to one of the following four cases:
\begin{description}
\item[case 1] $0 < A^*  < B^* < \infty$ with $B^* = \underline{b}(A^*) =  \overline{b}(A^*) < \infty$;
\item[case 2] $0 < A^* < B^* = \infty$ with $B^* = \underline{b}(A^*) =  \overline{b}(A^*) = \infty$ and $\widehat{\Psi}(A^*,\infty) = 0$;
\item[case 3] $0 = A^*  < B^* < \infty$ with $B^* = b(0) \leq \underline{b}(0)$;
\item[case 4] $0 = A^* < B^* = \infty$ with $\underline{b}(0) = \infty$ and $b(0) = \infty$.
\end{description}

\begin{theorem} \label{lemma_classification_b}
\begin{enumerate}
\item If $\underline{A} > 0$ and $\underline{b}(\underline{A}) < \infty$, then there exists $A^* \in (\underline{A},\overline{A})$ such that $B^* = \underline{b}(A^*) = \overline{b}(A^*) < \infty$.  This corresponds to \textbf{case 1}.
\item If $\underline{A} > 0$ and $\underline{b}(\underline{A}) = \infty$, then $A^*=\underline{A}$ and $B^*=\infty$ satisfy the condition for \textbf{case 2}.
\item If $\underline{A} = 0$, $\overline{A} > 0$, and  $\underline{b}(0) < \overline{b}(0)$, then there exists $A^* \in (0,\overline{A})$ such that $B^* = \underline{b}(A^*) = \overline{b}(A^*)$.  This corresponds to \textbf{case 1}.
\item Suppose (i) $\overline{A} = 0$ or (ii) $\underline{A} = 0$ and  $\underline{b}(0) \geq \overline{b}(0)$. If $b(0) < \infty$, then $A^*=0$ and $B^* = b(0)$ satisfy the condition for \textbf{case 3}. If $b(0) = \infty$, then $A^*=0$ and $B^* = \infty$ satisfy the condition for \textbf{case 4}.
\end{enumerate}
\end{theorem}

In particular,  from  \eqref{bounded_var_cond} and \eqref{psi_A_under_over} we infer that $\int_0^1 u \Pi (\diff u) = \infty$ implies $\underline{A} > 0$.  This together with  Theorem \ref{lemma_classification_b} leads to the following corollary.

\begin{corollary}  \label{A_0_unbounded_variation} If $X^d$ as in \eqref{x_decomposition} has paths of unbounded variation, then $\int_0^1 u \Pi (\diff u) = \infty$ and $A^* > 0$.
\end{corollary}

\begin{remark} \label{remark_A_B_sign}
Note that $\underline{b}(A) = \overline{b}(A)$ implies $b(A)=\underline{b}(A) = \overline{b}(A)$ (even when they are $+\infty$; see \eqref{eq_behavior_gamma}). By the construction in \eqref{def_b}, $A^*$ and $B^*$ obtained above must satisfy:
\begin{enumerate}
\item For every $A^* < B < B^*$, $\widehat{\Psi}(A^*,B) < 0$ and  $\widehat{\Psi}(A^*,B) - \widehat{\psi}(A^*,B) \frac {W^{(r)}(B-A^*)} {W^{(r)'}(B-A^*)} < 0$.
\item If $A^* > 0$, then $\widehat{\Psi}(A^*,B^*) = 0$\, (continuous or smooth fit at $A^*$ is satisfied).
\item $\widehat{\Psi}(A^*,B^*) - \widehat{\psi}(A^*,B^*) \frac {W^{(r)}(B^*-A^*)} {W^{(r)'}(B^*-A^*)} = 0$\,  (continuous and smooth fit at $B^*$ is satisfied, see \eqref{equivalence_psi_hat}).
\end{enumerate}
\end{remark}


In Theorem \ref{lemma_classification_b}-(1,3), we need to further  identify $(A^*, B^*)$. To this end, we first observe
\begin{lemma}\label{lemma_b_bar}
(1)\, $\underline{b}(A)$ increases in $A$ on $(\underline{A},\overline{A})$,  and (2)\, $\overline{b}(A)$ decreases in $A$ on $(\underline{A},\overline{A})$.
\end{lemma}


This lemma implies that (i) if $\overline{b}(A) > \underline{b}(A)$, then $A^*$ must lie on $(A,\overline{A})$ and (ii) if $\overline{b}(A) < \underline{b}(A)$, then $A^*$ must lie on $(\underline{A},A)$.  By Lemma \ref{lemma_b_bar} and Theorem \ref{lemma_classification_b}, the following algorithm, motivated by the bisection method, is guaranteed to output the pair $(A^*,B^*)$.  Here let $\varepsilon > 0$ be the error parameter.
\begin{description}
\item[Step 1] Compute $\underline{A}$ and $\overline{A}$.
\begin{description}
\item[Step 1-1] If (i) $\overline{A} = 0$ or (ii) $\underline{A} = 0$ and  $\underline{b}(0) \geq \overline{b}(0)$, then stop and conclude that this is \textbf{case 3} or \textbf{4} with $A^*=0$ and $B^* = b(0)$.
\item[Step 1-2] If $\underline{A} > 0$ and $\underline{b}(\underline{A}) = \infty$, then stop and conclude that this is \textbf{case 2} with $A^* = \underline{A}$ and $B^* = \infty$.
\end{description}
\item[Step 2]  Set $A = (\underline{A} + \overline{A})/2$.
\item[Step 3] Compute $\overline{b}(A)$ and $\underline{b}(A)$.
\begin{description}
\item[Step 3-1] If $|\overline{b}(A) - \underline{b}(A)| \leq \varepsilon$, then stop and conclude that this is \textbf{case 1} with $A^* = A$ and $B^* = \underline{b}(A)$ (or $B^* = \overline{b}(A)$).
\item[Step 3-2] If $|\overline{b}(A) - \underline{b}(A)| > \varepsilon$ and $\overline{b}(A) > \underline{b}(A)$, then set $\underline{A} = A$ and go back to \textbf{Step 2}.
\item[Step 3-3] If $|\overline{b}(A) - \underline{b}(A)| > \varepsilon$ and $\overline{b}(A) < \underline{b}(A)$, then set $\overline{A} = A$ and go back to \textbf{Step 2}.
\end{description}
\end{description}

\subsection{Verification of Equilibrium} \label{subsection_verification} We are now ready to prove  Theorems \ref{theorem_equilibrium}-\ref{theorem_equilibrium_zero}.  Our candidate value function for the Nash equilibrium  is given   by \eqref{vbx} and \eqref{v_limit} with $A^*$ and $B^*$ obtained by the procedure above. By Lemma \ref{lemma_delta_b},
\begin{align}
\begin{split}
v_{A^*,B^*}(x) &= \left\{ \begin{array}{ll} h(x), & x \geq B^* \\
h(x) + (v_{A^*,B^*}(x) - h(x)), & A^* < x < B^* \\
g(x), & x \leq A^*
\end{array}\right\} = - \Big(\frac{\pcheck}{r} +\acheck\Big) \zeta(x) + J(x)
\end{split} \label{value_function}
\end{align}
where
\begin{align} \label{def_J}
J(x) := \left\{ \begin{array}{ll} \frac{\pcheck}{r} - \costb, & x \geq B^*, \\
 \Upsilon(x; A^*, B^*),  & A^* < x < B^*, \\
\frac{\pcheck}{r} +\costs, & 0 < x \leq A^*, \\ \frac{\pcheck}{r} +\acheck & x \leq 0.
\end{array}\right.
\end{align}
When $A^* > 0$, $(\sigma_{A^*}, \tau_{B^*})$ is the candidate saddle point that attains $v_{A^*,B^*}(x)$.  When $A^*=0$, $v_{0+,B^*}(x)$ can be approximated by $(\sigma_{\varepsilon}, \tau_{B^*})$ for sufficiently small $\varepsilon > 0$.  The value of $\Upsilon(x; A^*, B^*)$ can be computed by \eqref{upsilon_in_terms_of_scale_function}, \eqref{upsilon_infty} and \eqref{upsilon_zero_via_scale}.

%


The proof of Theorems \ref{theorem_equilibrium}-\ref{theorem_equilibrium_zero} involves the crucial steps:
\begin{enumerate}[(i)]
\item Domination property
\begin{enumerate}
\item $\E^x \left[ e^{-r (\tau \wedge \sigma_{A^*})} v_{A^*,B^*}(X_{\tau \wedge \sigma_{A^*}}) 1_{\{\tau \wedge \sigma_{A^*} < \infty \}}\right] \geq v(x; \sigma_{A^*},\tau)$ for all $\tau \in \S$;
\item  $\E^x \left[ e^{-r (\sigma \wedge \tau_{B^*})} v_{A^*,B^*}(X_{\sigma \wedge \tau_{B^*}}) 1_{\{\sigma \wedge \tau_{B^*} < \infty \}}\right] \leq v(x;\sigma, \tau_{B^*})$ for all $\sigma \in \S$;
\end{enumerate}
\item Sub/super-harmonic property
\begin{enumerate}
\item $(\mathcal{L}-r) v_{A^*,B^*}(x) > 0$ for every $0 < x < A^*$;
\item $(\mathcal{L}-r) v_{A^*,B^*}(x) = 0$ for every $A^*< x < B^*$;
\item $(\mathcal{L}-r) v_{A^*,B^*}(x) < 0$ for every $x > B^*$.
\end{enumerate}
\end{enumerate} Here $\mathcal{L}$ is the infinitesimal generator associated with
the process $X$
\begin{align*}
\mathcal{L} f(x) &= c f'(x) + \frac 1 2 \nu^2 f''(x) + \int_0^\infty \left[ f(x-z) - f(x) +  f'(x) z 1_{\{0 < z < 1\}} \right] \Pi(\diff z)
\end{align*}
applied to any bounded and sufficiently smooth function $f$ that is $C^2$ when $X$ is of unbounded variation and $C^1$ otherwise.

After establishing (i)-(ii) above, we will apply them to establish \eqref{Nashv} by showing for the candidate optimal thresholds $(A^*,B^*)$ that
\begin{align}\label{NashvAB}
v(x; \sigma_{A^*}, \tau)  \leq v_{A^*,B^*}(x)\leq v(x; \sigma, \tau_{B^*}), \qquad \forall \sigma, \tau \in \S.
\end{align}

\begin{remark}\label{remark_laststep} In fact, it is sufficient to show \eqref{NashvAB}  holds  for all $\tau \in \S_{A^*}$ and $\sigma \in \S_{B^*}$, where
\begin{align}
\begin{split}
\S_{A^*} := \left\{ \tau \in \S:  X_{\tau} \notin (0,A^*]\; a.s.\right\} \quad \textrm{and} \quad \S_{B^*} := \left\{ \sigma \in \S:  X_{\sigma} \notin [B^*,\infty) \; a.s. \right\}.
\end{split} \label{set_stopping_times_modified}
\end{align}Indeed, for any candidate $\tau \in \S$, it follows that $v(x; \sigma_{A^*}, \tau)\leq v(x; \sigma_{A^*}, \hat{\tau})$ where $\hat{\tau} := \tau 1_{\{X_\tau \notin (0,A^*]\}} + {\sigma_0} 1_{\{X_\tau \in (0,A^*]\}}\in \S_{A^*}$, so the buyer's optimal exercise time $\tau^*$ must belong to $\S_{A^*}$. This is intuitive since  the seller will end the game as soon as $X$ enters $(0,A^*]$ and hence the buyer should not  needlessly stop in this interval and pay $\costb$.   Similar arguments apply to the use of $\S_{B^*}$. Then, using the same arguments as for \eqref{vbx}, we can again safely eliminate the $f (\cdot)$ term in  \eqref{definition_v} and write 
\begin{align*}
v(x; \sigma_{A^*}, \tau) &= \E^x \left[ e^{-r (\tau\wedge \sigma_{A^*})} \left( h(X_{\tau}) 1_{\{\tau < \sigma_{A^*} \}} + g(X_{\sigma_{A^*}}) 1_{\{\tau > \sigma_{A^*} \}} \right)  1_{\{\tau \wedge \sigma_{A^*} < \infty\}} \right], \quad \tau \in \S_{A^*}, \\
v(x; \sigma, \tau_{B^*}) &= \E^x \left[ e^{-r (\tau_{B^*}\wedge \sigma)} \left( h(X_{\tau_{B^*}}) 1_{\{\tau_{B^*} < \sigma \}} + g(X_{\sigma}) 1_{\{\tau_{B^*} > \sigma \}}  \right) 1_{\{\tau_{B^*} \wedge \sigma < \infty\}} \right], \quad \sigma \in \S_{B^*}.
\end{align*}
\end{remark}

We prove properties (i)-(ii) above using the following lemmas.

\begin{lemma} \label{lemma_domination}
For every $x \in (A^* , B^*)$, the following inequalities hold:
\begin{align}
\label{ineq_g}v_{A^*, B^*}(x) - g(x) &\leq 0, \\
\label{ineq_h}v_{A^*, B^*}(x) - h(x) &\geq 0,
\end{align}
where it is understood for the case $A^*=0$ and $\nu > 0$ that the above results hold with $A^* = 0+$.
\end{lemma}

Applying this lemma and the definitions of $\S_{A^*}$ and $\S_{B^*}$ in \eqref{set_stopping_times_modified} of Remark \ref{remark_laststep}, we obtain
\begin{lemma} \label{lemma_domination2}
Fix $x > 0$.
\begin{enumerate}
\item For every $\tau \in \S_{A^*}$, when $A^* > 0$
\begin{align*}
g(X_{\sigma_{A^*}}) 1_{\{\sigma_{A^*} < \tau \}} + h(X_{\tau}) 1_{\{\tau < \sigma_{A^*}\}} \leq v_{A^*,B^*}(X_{\sigma_{A^*} \wedge \tau}), \quad \p^x\textrm{-a.s.} \; \textrm{on } \{\sigma_{A^*} \wedge \tau < \infty\},
\end{align*}
and when $A^* = 0$,
\begin{align*}
- (\acheck - \gamma_s)  1_{\{X_{\tau} = 0\}} + h(X_{\tau}) 1_{\{\tau < {\sigma_0} \}} \leq v_{0+,B^*}(X_{\tau}), \quad \p^x\textrm{-a.s.} \; \textrm{on } \{\tau < \infty\}.
\end{align*}

\item For every $\sigma \in \S_{B^*}$,
\begin{align*}
g(X_{\sigma}) 1_{\{\sigma < \tau_{B^*} \}} + h(X_{\tau_{B^*}}) 1_{\{\tau_{B^*} < \sigma \}} \geq v_{A^*,B^*}(X_{\sigma \wedge \tau_{B^*}}), \quad \p^x\textrm{-a.s.} \; \textrm{on } \{\sigma \wedge \tau_{B^*} < \infty\},
\end{align*}
where it is understood for the case $A^*=0$ and $\nu > 0$ that the above holds with $A^* = 0+$.
\end{enumerate}
\end{lemma}

\begin{lemma} \label{lemma_generator}
\begin{enumerate}
\item When $A^* > 0$, we have $(\mathcal{L}-r) v_{A^*,B^*}(x) > 0$ for every $0 < x < A^*$.
\item We have  $(\mathcal{L}-r) v_{A^*,B^*}(x) = 0$ for every $A^* < x < B^*$.
\item When $B^* < \infty$, we have $(\mathcal{L}-r) v_{A^*,B^*}(x) < 0$ for every $ x > B^*$.
\end{enumerate}
\end{lemma}

The domination property (i) holds by applying discounting and expectation in Lemma \ref{lemma_domination2}.  The sub/super-harmonic property (ii) is implied by Lemma \ref{lemma_generator}. By Ito's lemma, this shows that the stopped processes $e^{-r(t \wedge \sigma_{A^*})}v_{A^*,B^*}(X_{t \wedge \sigma_{A^*}})$ and $e^{-r(t \wedge \tau_{B^*})}v_{A^*,B^*}(X_{t \wedge \tau_{B^*}})$ are, respectively, a supermartingale and  a submartingale.  In turn, we apply them to show $v_{A^*,B^*}(x) \geq v(x; \sigma_{A^*},\tau)$ for any arbitrary $\tau \in \S_{A^*}$, and $v_{A^*,B^*}(x) \leq v(x; \sigma,\tau_{B^*})$ for any arbitrary $\sigma \in \S_{B^*}$, that is, the Nash equilibrium. We provide the details of the proofs for Theorems \ref{theorem_equilibrium}-\ref{theorem_equilibrium_zero} in the Appendix.

\section{Exponential Jumps and Numerical Examples}\label{section_numer}
In this section, we consider   spectrally negative \lev processes with i.i.d.\ exponential jumps and provide some numerical examples to  illustrate the buyer's and seller's optimal exercise strategies and the impact of step-up/down fees on the game value.  The results obtained here can be extended easily to the hyperexponential case using the explicit expression of the scale function obtained by \cite{Egami_Yamazaki_2010_2}, and can be used to approximate for a general case with a completely monotone density (see, e.g., \cite{Egami_Yamazaki_2010_2,Feldmann_1998}).  Here, however, we focus on a rather simple case for more intuitive interpretation of our numerical results. 


\subsection{Spectrally Negative \lev Processes with Exponential Jumps}
Let $X$ be a spectrally negative \lev process of the form
\begin{equation}
  X_t  - X_0=\mu t+\nu B_t - \sum_{n=1}^{N_t} Z_n, \quad 0\le t <\infty. \notag
\end{equation}
Here $B=\{B_t; t\ge 0\}$ is a standard Brownian motion, $N=\{N_t; t\ge 0\}$ is a Poisson process with arrival rate $\lambda$, and  $Z = \left\{ Z_n; n = 1,2,\ldots \right\}$ is an i.i.d.\ sequence of exponential   random variables with density function $f (z)  :=  \eta e^{- \eta z}$, $z > 0$,
for some $0 < \eta < \infty$.  Its Laplace exponent (\ref{laplace_spectrally_negative}) is given by
\begin{align*}
\phi(s) = \mu s + \frac 1 2 \nu^2 s^2  - \lambda
\frac s {\eta + s}.
\end{align*}

For our examples,  we assume $\nu > 0$.  In this case, there are two \emph{negative} solutions to the equation $\phi(s) = r$ and their absolute values $\{\xi_{i,r}; i = 1, 2\}$ satisfy the interlacing condition: $0 < \xi_{1,r} < \eta < \xi_{2,r}  < \infty$.
For this process,  the scale function is given by for every $x \geq 0$
\begin{align}
\begin{split}
W^{(r)}(x) &=  \sum_{i = 1}^{2} C_i \left[ e^{\Phi(r) x} - e^{-\xi_{i,r}x} \right],
\end{split} \label{scale_hyperexponential}
\end{align}
for some $C_1$ and $C_2$ (see \cite{Egami_Yamazaki_2010_2} for their expressions).
In addition, applying \eqref{scale_hyperexponential} to \eqref{w_phi} yields
\begin{align*}
W_{\Phi(r)}(x) &=  \sum_{i = 1}^{2} C_i \left[ 1 - e^{-(\Phi(r)+\xi_{i,r})x} \right],
\end{align*}
with the limit $W_{\Phi(r)}(\infty) = \sum_{i=1}^{2} C_i$, which equals $(\phi'(\Phi(r)))^{-1}$ by \eqref{scale_function_asymptotic}.


 Recall that, in contrast to $\psi(A,B)$ and $\Psi(A,B)$,  $\widehat{\psi}(A,B)$ and $\widehat{\Psi}(A,B)$ do not explode. Therefore, they are used to compute  the optimal thresholds $A^*$ and $B^*$ and the value function $V$. Below we provide the formulas for $\widehat{\psi}(A,B)$ and $\widehat{\Psi}(A,B)$. The computations are very tedious but straightforward, so we omit the proofs here.

In summary, for  $B > A\geq0$, we have
\begin{multline*}
\widehat{\psi}(A,B)
=  -\left( \pcheck + \costs r \right) + \left( \acheck - \costs \right) \lambda e^{-\eta A} - \frac { W_{\Phi(r)} (\infty)} {W_{\Phi(r)} (B-A)} (\acheck - \costs) \lambda   \frac {\eta} {\Phi(r)+\eta}  e^{- \eta A} \\
+ \frac { \acheck - \costs } {W_{\Phi(r)} (B-A)} \lambda  e^{-\Phi(r)(B-A)} \sum_{i = 1}^{2}  C_i \Big[ \frac {\eta} {\Phi(r)+\eta} e^{-\eta B} + \frac {\eta} {\xi_{i,r}-\eta} \left( e^{-\eta B } - e^{-\xi_{i,r}(B-A)- \eta A}\right)\Big]
\end{multline*}and
\begin{align*}
\widehat{\Psi}(A,B) = \frac 1 {W_{\Phi(r)}(B-A)} \Big[    W_{\Phi(r)}(\infty)  \Big( \lambda e^{-\eta A}  \frac {\acheck - \costs } {\Phi(r)+\eta} - \frac {\pcheck  +\costs r} {\Phi(r)} \Big) + {e^{-\Phi(r) (B-A)}} \varrho(A,B) \Big]
\end{align*}
where
\begin{multline*}
\varrho(A,B) :=  \left(\acheck - \costs \right) \lambda  \sum_{i = 1}^{2} C_i  \Big( e^{-\eta B} \Big[ - \frac 1 {\Phi(r)+\eta}  + \frac 1 {\xi_{i,r}-\eta}  \Big] - e^{-\eta A - \xi_{i,r}(B-A)} \frac 1 {\xi_{i,r}-\eta}\Big) \\
- \left( \pcheck + \costs r\right) \sum_{i = 1}^{2} C_i  \Big[ -\frac 1 {\Phi(r)}   + \frac 1 {\xi_{i,r}} \left( e^{-\xi_{i,r}(B-A)} - 1  \right) \Big] - \left( \costb + \costs\right).
\end{multline*}
Also, setting $B=\infty$ and $B=A+$, \eqref{psi_A_under_over}-\eqref{psi_A_under_over2} yield
\begin{align*}
\widehat{\psi}(A,\infty) =  -\left( \pcheck + \costs r \right) +  \lambda   (\acheck - \costs)\frac {\Phi(r)} {\Phi(r)+\eta}  e^{- \eta A} \quad \textrm{and} \quad
\widehat{\psi}(A,A+) =  -\left( \pcheck + \costs r \right)  + \left( \acheck - \costs \right) \lambda   e^{-\eta A}.
\end{align*}

\subsection{Numerical Results}

Let us denote the step-up/down ratio by $q:=\hat{p}/p = \hat{\alpha} / \alpha.$   We consider     four contract specifications:
\begin{enumerate}
\item[(C)] cancellation game with $q=0$\,\, (position canceled at exercise),
\item[(D)] step-down game with $q=0.5$\,\, (position halved at exercise),
\item[(V)] vanilla CDS with $q=1.0$\,\, (position unchanged at exercise),
\item[(U)] step-up game with $q = 1.5$\,\, (position raised at exercise).
\end{enumerate}
The model parameters are $r =0.03$, $\lambda = 1.0$, $\eta = 2.0$, $\nu = 0.2$,  $\alpha=1$, $x=1.5$ and $\gamma_s = \gamma_b = 1000$\,bps, unless specified otherwise. We also choose $\mu$ so that the risk-neutral condition $\phi(1) = r$ is satisfied.

Figure \ref{figure_value_x} shows for all four cases the contract value $V$ to the buyer as a function of $x$ given a fixed premium rate.  It is decreasing in $x$  since   default  is less likely for higher value of $x$.  For the cancellation game, $V$ takes the constant values $\costs=1000$\,bps for $x\leq A^*$ and $-\costb=-1000$\,bps for $ x\geq B^*$ since in these regions immediate cancellation with a fee is optimal.
\begin{figure}[htbp]
\begin{center}
\begin{minipage}{1.0\textwidth}
\centering
\begin{tabular}{c}
\includegraphics[scale=0.58]{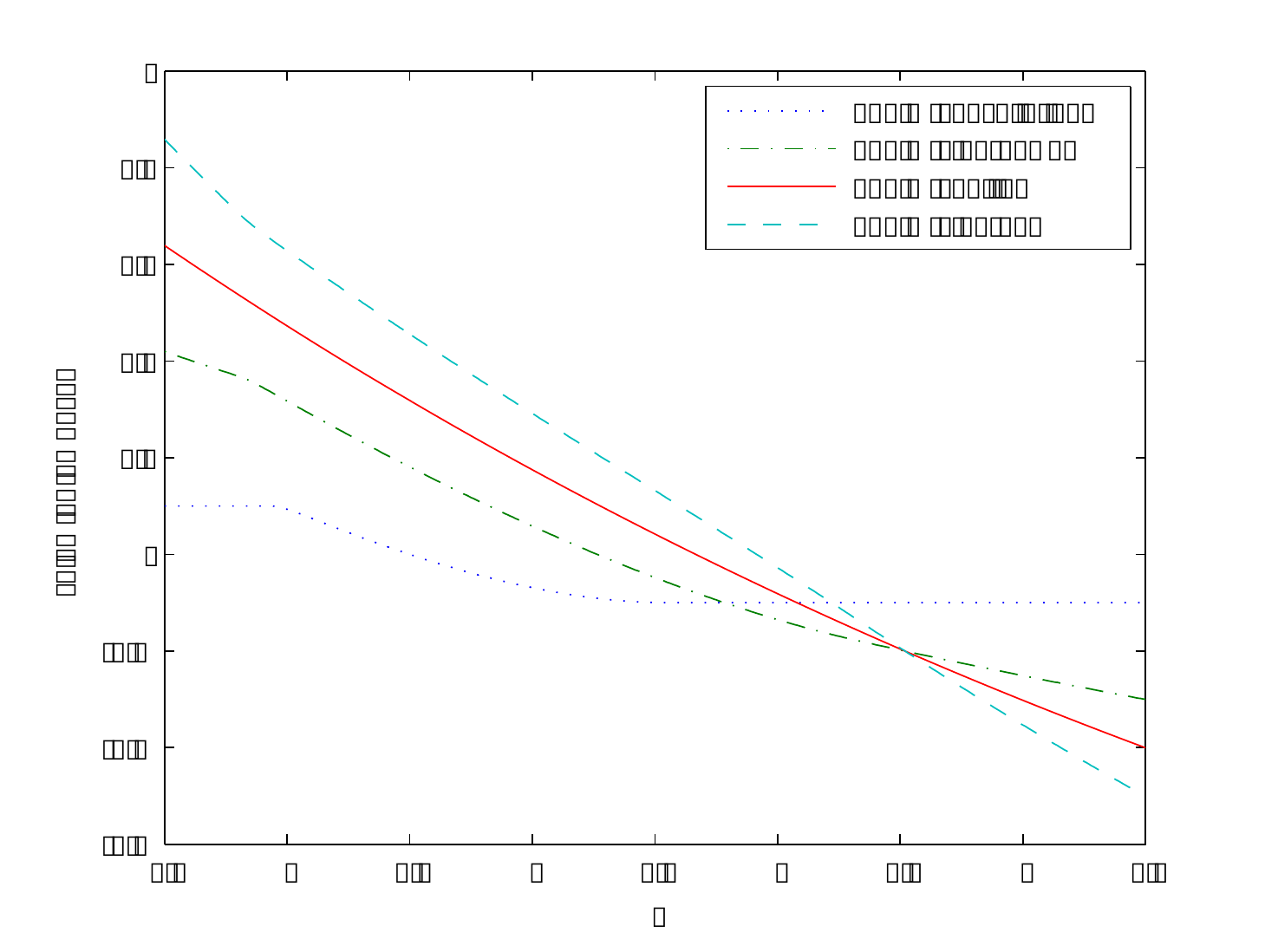}
\end{tabular}
\end{minipage}
\caption{\small The value for the buyer $V(x;\sigma_{A^*},\tau_{B^*})$ as a function of $x$. Here $r=0.03$, $p=0.05$, $\mu = 0.1352$, $\lambda = 1.0$, $\eta = 2.0$, $\nu = 0.2$, and $\gamma_b = \gamma_s = 1000$\,bps.} \label{figure_value_x}
\end{center}
\end{figure}

\begin{figure}[ht]
\centering
\includegraphics[scale=0.58]{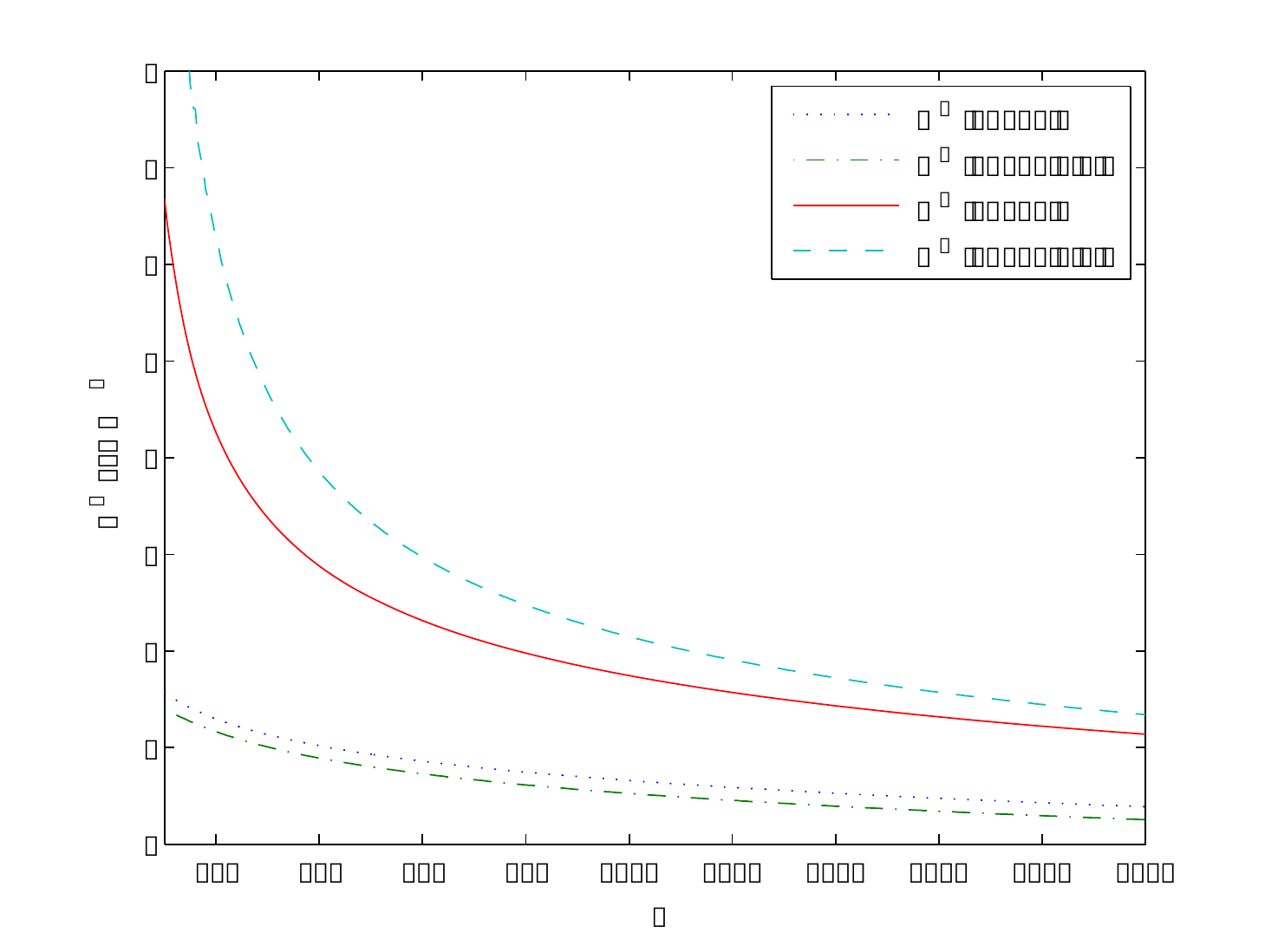}
   \includegraphics[scale=0.58]{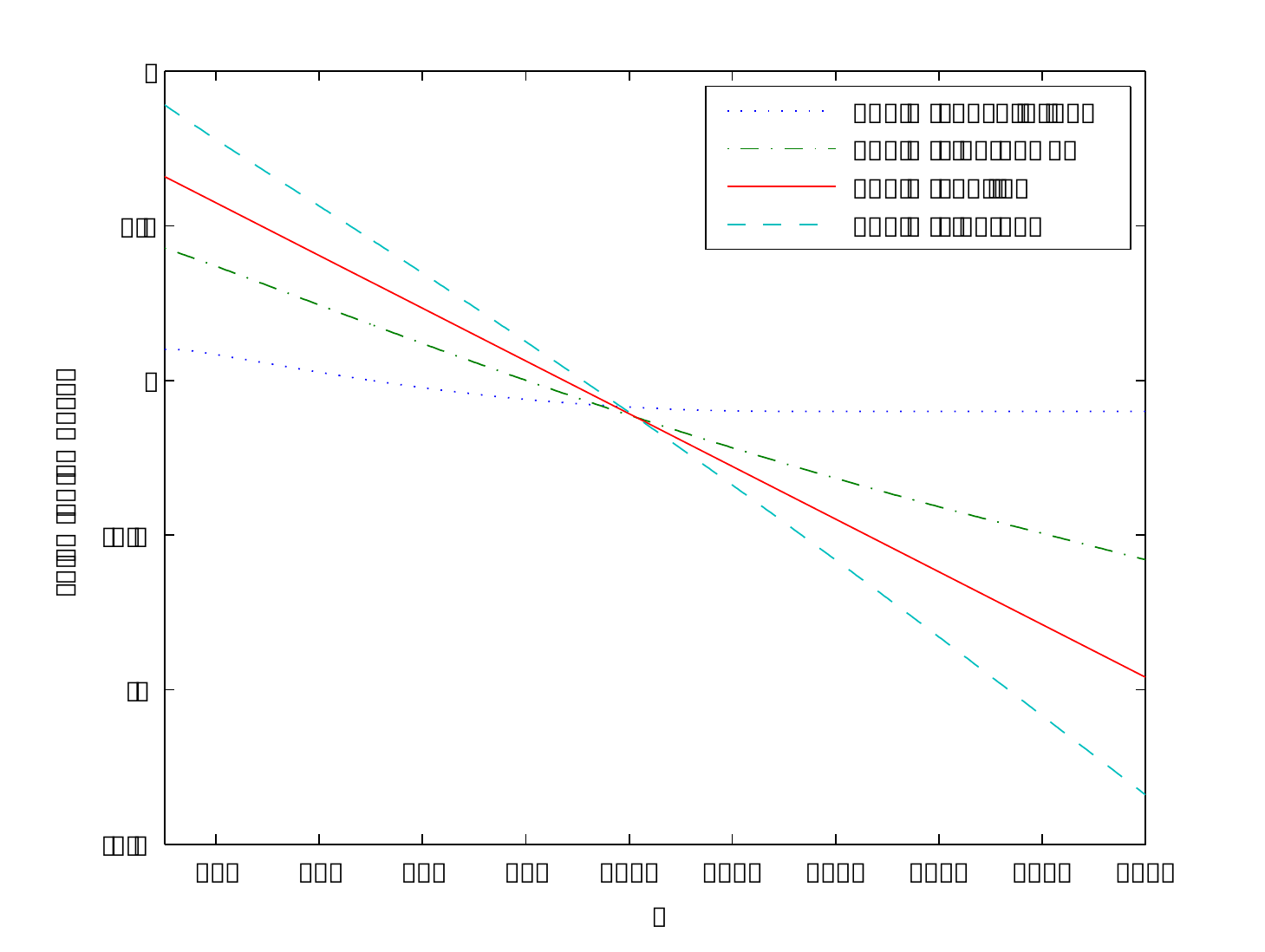}
\caption{\small (Left) Optimal threshold levels $A^*$ and $B^*$ and (right) the value for the buyer with respect to $p$. The parameters are $r=0.03$, $x=1.5$, $\mu = 0.3433$, $\lambda = 0.5$, $\eta = 2.0$, $\nu = 0.2$, and $\gamma_b = \gamma_s = 1000$\,bps.}\label{figure_wrt_p}
\end{figure}

\begin{figure}
\centering
\includegraphics[scale=0.5]{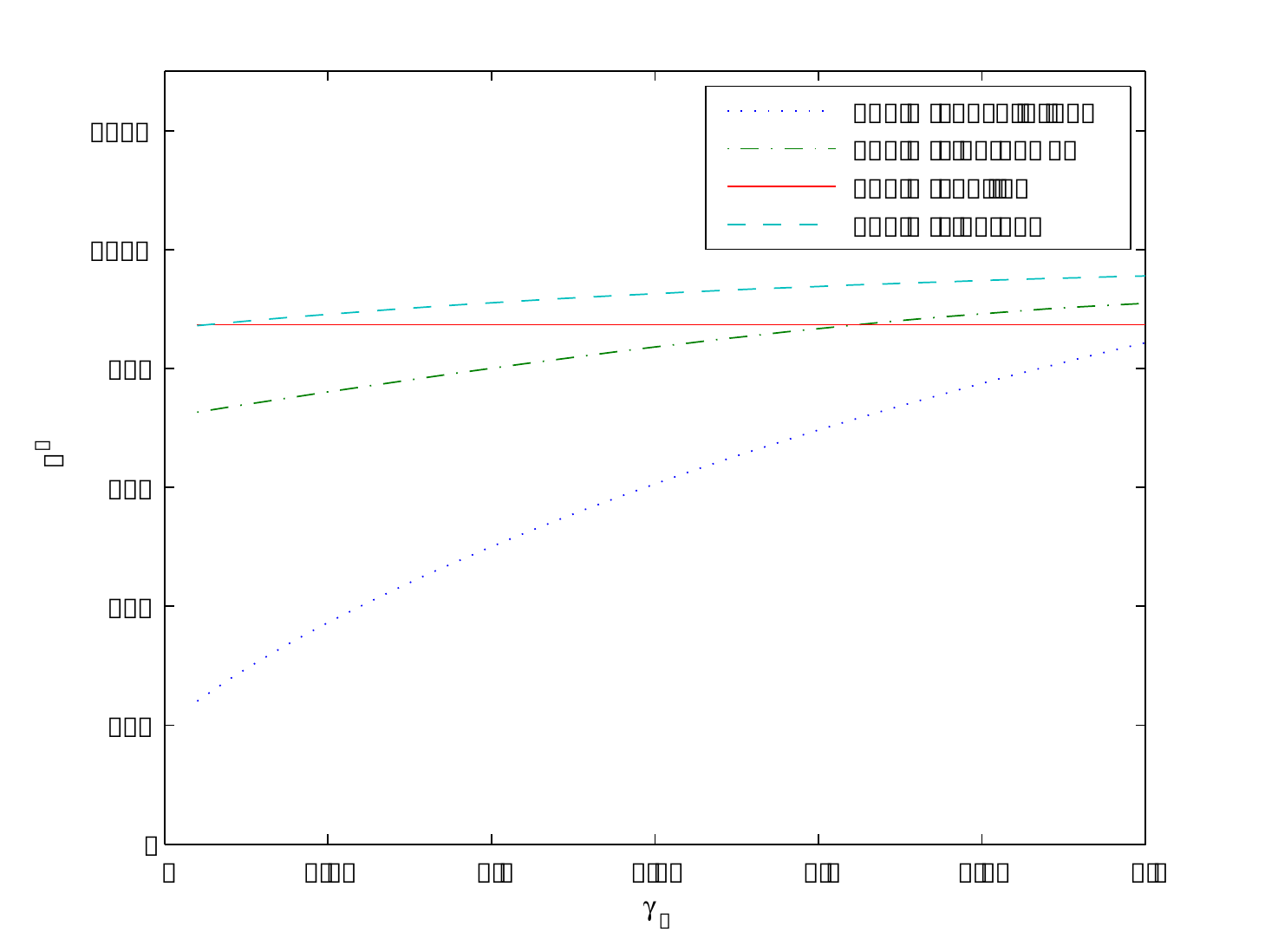}
\includegraphics[scale=0.5]{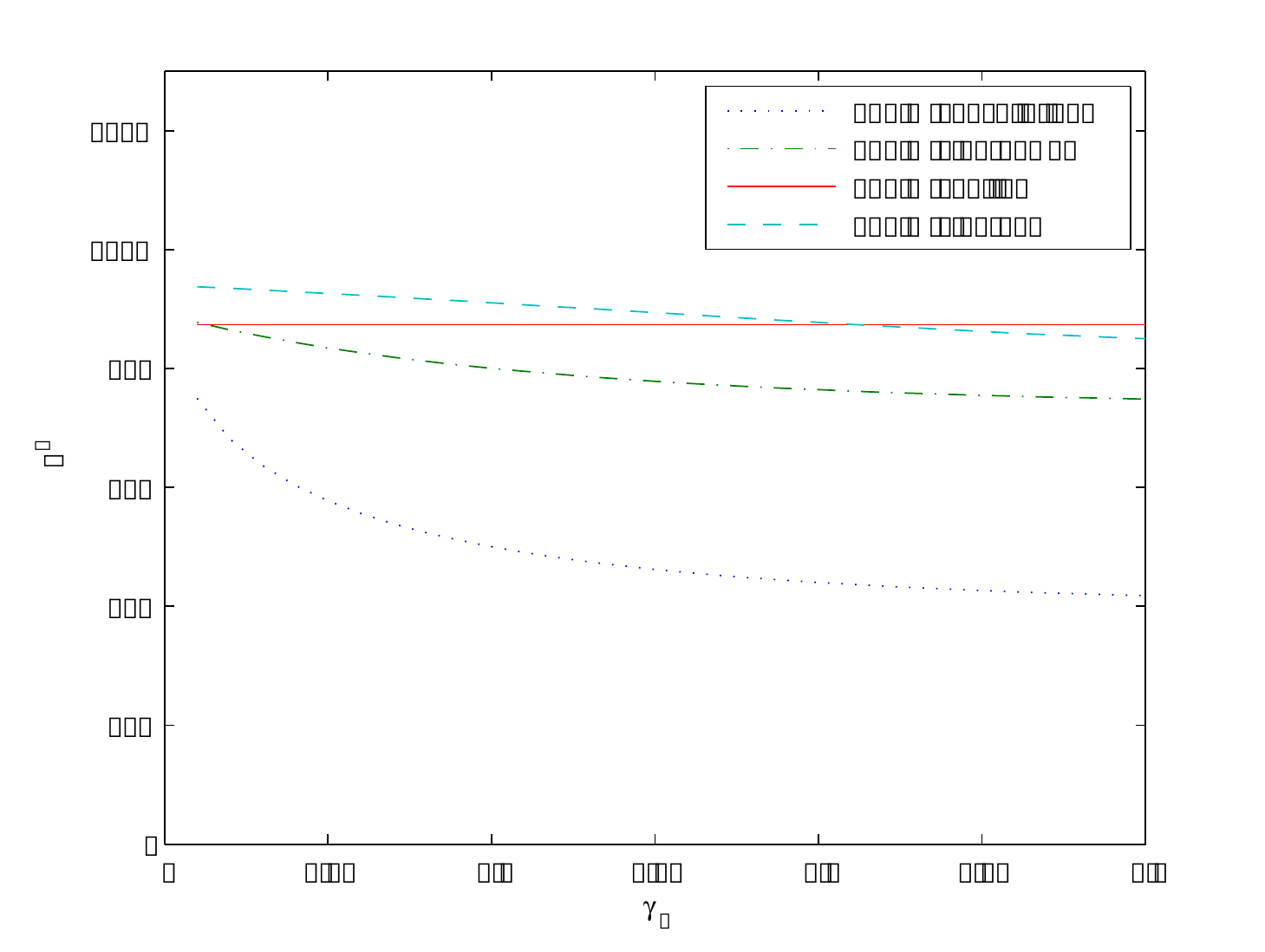}
\caption{\small The equilibrium premium $p^*$ with respect to $\gamma_s$ (left) and $\gamma_b$ (right). Here $r=0.03$, $x=1.5$, $\mu = 0.3433$, $\lambda = 1.0$, $\eta = 2.0$, $\nu = 0.2$, and $\gamma_b = \gamma_s = 1000$\,bps unless specified otherwise.} \label{figure_wrt_gamma}
\end{figure}

In Figure \ref{figure_wrt_p}, we show the optimal thresholds $A^*$ and $B^*$ and the value $V$ with respect to $p$.  The symmetry argument discussed in Section \ref{section_model} applies to the cases   (D)  and (U).   As a result, the $A^*$ in (D) is identical to the $B^*$ in (U), and the $B^*$ in (D) is identical to the $A^*$ in (U).   In all four cases, both $A^*$ and $B^*$ are decreasing in $p$. In other words, as $p$ increases, the buyer tends to exercise earlier  while the seller tends to delay exercise. Intuitively, a higher premium makes waiting more costly for the buyer but more profitable for the seller. The value $V$ in the cancellation game stays constant when $p$ is sufficiently small because the seller would exercise immediately; it also becomes flat when $p$ is sufficiently high because the buyer would exercise immediately.

Note that the value function $V$ and the optimal stopping strategies $(\sigma^*, \tau^*)$ depend on the premium rate $p$.  In particular, we call $p^*$ the \emph{equilibrium premium} rate if it yields $V(x;\sigma^*(p^*),\tau^*(p^*)) = 0$, where we emphasize the  saddle point $(\sigma^*(p^*),\tau^*(p^*))$ corresponds  to $p^*$. Hence, under the equilibrium premium rate, the default swap game starts at value zero, implying no cash transaction between the protection buyer and seller at contract initiation.
As illustrated in  Figure \ref{figure_wrt_p}-(b), the value $V$ (from the buyer's perspective) is always decreasing in $p$. Using a  bisection method, we numerically determine the equilibrium premium $p^*$ so that $V=0$.
We illustrate in Figure \ref{figure_wrt_gamma}  the equilibrium premium $p^*$ as a function of $\gamma_b$ and $\gamma_s$.
As is intuitive, the equilibrium premium $p^*$ is   increasing in $\gamma_s$ and decreasing in $\gamma_b$.
\section{Conclusions}\label{section_conclude}
We  have discussed the valuation of a default swap contract where the protection buyer and seller can alter the respective position  once prior to default.  This new contractual feature  drives  the protection buyer/seller to consider the optimal timing to control credit risk exposure. The valuation problem involves the analytical and numerical studies of  an optimal stopping game with early termination from default. Under a perpetual setting, the investors' optimal stopping rules are characterized by their respective exercise  thresholds, which can be quickly determined in a general class of spectrally negative \lev credit risk models.

For future research, it is most natural to consider the default swap game under a finite horizon and/or different credit risk models. The default swap game  studied in this paper can potentially be applied to approximate its finite-maturity version  using the maturity randomization (Canadization) approach (see \cite{carr_Canada,kypri_Canada}).  Another interesting extension  is to allow for multiple adjustments by the buyer and/or seller prior to default, which  can be modeled as stochastic games with multiple stopping opportunities. Finally,  the step-up and step-down features also arise in other derivatives, including interest rate swaps.

\appendix
\section{Proofs}
\begin{proof}[Proof of Proposition \ref{prop-V}]
First, by a rearrangement of integrals and \eqref{ap}, the expression inside
the expectation in \eqref{def_V} can be written as
\begin{align*}
&1_{\{ \tau \wedge \sigma < \infty \}} \Big[   \int_{\tau\wedge \sigma}^{\sigma_0} e^{-rt} \pcheck\, \diff t - \int_0^{\sigma_0} e^{-rt} p\, \diff t  + e^{-r {\sigma_0}} \left( -\acheck 1_{\{\tau\wedge \sigma < {\sigma_0}\}}  + \alpha \right) \\ &\qquad + 1_{\{\tau \wedge \sigma < {\sigma_0} \}} e^{-r (\tau \wedge \sigma)}  \left( -\costb 1_{\{\tau \leq \sigma \}} + \costs 1_{\{\tau \geq \sigma \}} \right) \Big] + 1_{\{ \tau \wedge \sigma = \infty \}} \Big( - \int_0^\infty e^{-rt} p\, \diff t\Big) \\
&= 1_{\{ \tau \wedge \sigma < \infty \}} \Big(\int_{\tau\wedge \sigma}^{\sigma_0} e^{-rt} \pcheck \,\diff t - e^{-r {\sigma_0}} \acheck 1_{\{\tau\wedge \sigma < {\sigma_0}\}} + 1_{\{\tau \wedge \sigma < {\sigma_0} \}} e^{-r (\tau \wedge \sigma)}  \left( -\costb 1_{\{\tau \leq \sigma \}} + \costs 1_{\{\tau \geq \sigma \}} \right)  \Big)  \\ &- \int_0^{{\sigma_0}} e^{-rt} p\, \diff t + e^{-r {\sigma_0}} \alpha \\
&= 1_{\{ \tau \wedge \sigma < \infty, \; \tau\wedge \sigma < {\sigma_0} \}} \Big(\int_{\tau\wedge \sigma}^{\sigma_0} e^{-rt} \pcheck \,\diff t - e^{-r {\sigma_0}} \acheck +  e^{-r (\tau \wedge \sigma)}  \left( -\costb 1_{\{\tau \leq \sigma \}} + \costs 1_{\{\tau \geq \sigma \}} \right) \Big)  - \int_0^{{\sigma_0}} e^{-rt} p\, \diff t + e^{-r {\sigma_0}} \alpha.
\end{align*}
Taking expectation, \eqref{def_V} simplifies to
\begin{align}\label{v_expectation}V(x; \sigma, \tau) &= \E^x \Big[1_{\{ \tau \wedge \sigma < \infty, \; \tau\wedge \sigma < {\sigma_0} \}} \Big(\int_{\tau\wedge \sigma}^{\sigma_0} e^{-rt} \pcheck \,\diff t - e^{-r {\sigma_0}} \acheck + e^{-r (\tau \wedge \sigma )} \left( -\costb 1_{\{\tau \leq \sigma \}} +  \costs 1_{\{\tau \geq \sigma \}} \right) \Big) \Big]\\ &-\E^x \Big[ \int_0^{\sigma_0} e^{-rt} p\, \diff t \Big]  +\alpha\, \E^x \left[ e^{-r {\sigma_0}}  \right]. \notag
\end{align}
Here, the last two terms   do not depend on $\tau$ nor $\sigma$ and they
constitute $C(x;p,\alpha) $. Next,  using the fact that $\{ \tau \wedge
\sigma  < {\sigma_0},  \tau \wedge \sigma < \infty \} =\{ X_{\tau \wedge
\sigma} >  0, \; \tau \wedge \sigma < \infty \}$ for every $\tau,\sigma \in
\mathcal{S}$ and the strong Markov property of $X$ at time $\tau \wedge
\sigma$, we express the first term as
\begin{align}
&\E^x \Big[ 1_{\{ \tau \wedge \sigma < \infty, \; \tau\wedge \sigma < {\sigma_0} \}} \Big( \E^x \Big[  \int_{\tau\wedge \sigma}^{\sigma_0} e^{-rt} \pcheck \,\diff t - e^{-r {\sigma_0}} \acheck \Big| \mathcal{F}_{\tau \wedge \sigma}\Big]  + e^{-r(\tau \wedge \sigma)}\left( -  \costb 1_{\{\tau \leq \sigma \}} + \costs 1_{\{\tau \geq \sigma \}} \right) \Big) \Big] \notag \\
&= \E^x \left[ 1_{\{ \tau \wedge \sigma < \infty, \; \tau\wedge \sigma < {\sigma_0} \}} e^{-r(\tau \wedge \sigma)} \left( h(X_{\tau}) 1_{\{\tau < \sigma\}} + g(X_{\sigma}) 1_{\{\tau > \sigma\}} + f(X_{\tau}) 1_{\{\tau = \sigma\}} \right) \right] \notag\\
&= \E^x \left[  e^{-r(\tau \wedge \sigma)} \left( h(X_{\tau}) 1_{\{\tau < \sigma\}} + g(X_{\sigma}) 1_{\{\tau > \sigma\}} + f(X_{\tau}) 1_{\{\tau = \sigma \}}  \right) 1_{\{ \tau \wedge \sigma < \infty \}}\right]=v(x; \sigma, \tau), \notag
\end{align}
where the second  equality holds because (i) $\tau < \sigma$ or $\tau > \sigma$ implies $\tau \wedge \sigma < {\sigma_0}$, and  (ii) by $f(X_{\sigma_0})=0$ we have $f(X_{\tau}) 1_{\{\tau = \sigma, \tau \wedge \sigma < {\sigma_0} \}}=f(X_{\tau}) 1_{\{\tau = \sigma\}}$ a.s. \end{proof}

\begin{proof}[Proof of Proposition \ref{prop-sym}]
First, we deduce from \eqref{hx}-\eqref{fx} that
\begin{align*}h(x; \pcheck, \acheck, \costb) &= -g(x;- \pcheck, -\acheck, \costb), \\
g(x; \pcheck, \acheck, \costs) & = -h(x;- \pcheck, -\acheck, \costs), \\
f(x; \pcheck, \acheck, \costb, \costs) &=  -f(x;- \pcheck, -\acheck, \costs, \costb).
\end{align*}
Substituting these equations to \eqref{definition_v} of Proposition \ref{prop-V}, it follows, for every $\tau,\sigma \in \S$, that
 \begin{align*}v(x; \sigma, \tau; \pcheck,  \acheck, \costb, \costs)  &= -\E^x \left[ e^{-r (\tau\wedge \sigma)} \left( h(X_{\sigma}; -\pcheck, -\acheck, \costs) 1_{\{\sigma < \tau \}} + g(X_{\tau}; -\pcheck, -\acheck, \costb) 1_{\{\tau < \sigma \}}  \right. \right. \\ & \qquad \left. \left. + f(X_{\tau\wedge \sigma}; -\pcheck, -\acheck, \costs, \costb) 1_{\{\tau = \sigma \}}\right) 1_{\{\tau \wedge \sigma < \infty\}}\right]\\
 &= -  v(x;  \tau,\sigma; -\pcheck, -\acheck, \costs, \costb) .\end{align*} \end{proof}

\begin{proof}[Proof of Lemma \ref{lemma_delta_b}]
Recall  that $v$ is given by the first expectation of \eqref{v_expectation}, and note that $\sigma_A \wedge \tau_B = \infty$ implies  ${\sigma_0} = \infty$. For every $x \in (A,B)$, $v(x;A,B) - h(x)$ equals
\begin{align*}
&\E^x\Big[ 1_{\{\sigma_A \wedge \tau_B < \infty\}} \Big(\int_{\sigma_A \wedge \tau_B}^{{\sigma_0}}
e^{-rt}\pcheck \,\diff t  - e^{-r {\sigma_0}}\acheck 1_{\{\sigma_A \wedge \tau_B <{\sigma_0}\}} + e^{-r (\sigma_A \wedge \tau_B)} \left(-\costb 1_{\{\tau_B < \sigma_A\}} + \costs 1_{\{\tau_B > \sigma_A \}} \right) \Big) \Big] \\ & \qquad - \E^x \Big[ \int_0^{\sigma_0} e^{-rt} \pcheck \, \diff t\,  - e^{-r {\sigma_0}} \acheck \, \Big] + \costb \\
&= \E^x\Big[ 1_{\{\sigma_A \wedge \tau_B < \infty\}} \Big(-\int_0^{\sigma_A \wedge \tau_B}
e^{-rt}\pcheck \,\diff t  + e^{-r {\sigma_0}}\acheck 1_{\{\sigma_A \wedge \tau_B = {\sigma_0}\}} + e^{-r (\sigma_A \wedge \tau_B)} \left(-\costb 1_{\{\tau_B < \sigma_A\}} + \costs 1_{\{\tau_B > \sigma_A \}} \right) \Big) \\ &\qquad  - 1_{\{\sigma_A \wedge \tau_B = \infty\}} \Big( \int_0^{\sigma_0} e^{-rt} \pcheck \, \diff t\,  - e^{-r {\sigma_0}} \acheck \Big) \, \Big] + \costb \\
&= \E^x\Big[  1_{\{\sigma_A \wedge \tau_B < \infty\}} e^{-r (\sigma_A \wedge \tau_B)}\left( \acheck 1_{\{\sigma_A \wedge \tau_B = {\sigma_0}\}} - \costb 1_{\{\tau_B < \sigma_A\}} + \costs 1_{\{\tau_B > \sigma_A \}} \right) - \int_0^{\sigma_A \wedge \tau_B} e^{-rt} \pcheck \, \diff t\, \Big] + \costb,
\end{align*}
which equals $\Upsilon(x;A,B) - \frac {\pcheck} r + \costb$.
Since $g(x) = h(x) + \costs + \costb$, $\forall x > 0$, the second
claim of \eqref{delta_by_upsilon}  is immediate.

The proof of the second claim amounts to proving the following: for  $0 < A < x < B < \infty$,
\begin{align}
\begin{split}
\E^x\left[  e^{-r (\sigma_A \wedge \tau_B)} 1_{\{\tau_B < \sigma_A \}} \right]  &= \frac {W^{(r)}(x-A)} {W^{(r)}(B-A)}, \\
\E^x\left[  e^{-r (\sigma_A \wedge \tau_B)} 1_{\{\tau_B > \sigma_A \, \textrm{or} \;  \sigma_A \wedge \tau_B = {\sigma_0} \}} \right] &=  Z^{(r)} (x-A) - Z^{(r)} (B-A) \frac {W^{(r)}(x-A)} {W^{(r)}(B-A)},\\
\E^x\left[  e^{-r (\sigma_A \wedge \tau_B)}1_{\{\sigma_A \wedge \tau_B = {\sigma_0} \}} \right] &=  \frac {W^{(r)}(x-A)} {W^{(r)}(B-A)}  \kappa(B;A) - \kappa(x;A).
\end{split} \label{about_lambda}
\end{align}
The first two equalities follow directly from  the
property of the scale function (see, for example, Theorem 8.1 of
\cite{Kyprianou_2006}).  Notice here that $\tau_B < \sigma_A$ if and only if it up-crosses $B$ before down-crossing $A$ while $\tau_B > \sigma_A$ or $\sigma_A \wedge \tau_B = {\sigma_0}$ if and only if it down-crosses $A$ before up-crossing $B$.
%


 For the third equality, we require the overshoot distribution that is again obtained via the scale function. Let  $N$ be the Poisson random measure for the jumps of $-X$ and $\overline{X}$ and $\underline{X}$ be the running maximum and minimum, respectively, of $X$. By compensation formula (see e.g.\ Theorem 4.4 of \cite{Kyprianou_2006}), we have
\begin{align}
\begin{split}
\E^x\left[  e^{-r (\sigma_A \wedge \tau_B)}1_{\{\sigma_A \wedge \tau_B = {\sigma_0} \}} \right]
&= \E^x \Big[ \int_0^\infty  \int_0^\infty N(\diff t \times \diff u) e^{ -r t} 1_{\{\overline{X}_{t-} < B, \; \underline{X}_{t-} > A, \; X_{t-} - u < 0 \}}\Big] \\
&= \E^x \Big[ \int_0^\infty \diff t e^{ -r t}  \int_0^\infty \Pi (\diff u) 1_{\{\overline{X}_{t-} < B, \; \underline{X}_{t-} > A, \; X_{t-} - u < 0 \}}\Big] \\ &= \int_0^\infty \Pi (\diff u) \int_0^\infty \diff t \left[ e^{ -r t} \p^x \{X_{t-} < u, \sigma_A \wedge \tau_B \geq t \} \right].
\end{split} \label{lambda_3_construction}
\end{align}
Recall that, as in Theorem 8.7 of \cite{Kyprianou_2006}, the resolvent measure for the spectrally negative \lev process killed upon exiting $[0,a]$ is given by
\begin{align*}
\int_0^\infty \diff t  \left[ e^{-rt}\p^x \left\{ X_{t-} \in \diff y, \sigma_0 \wedge \tau_a > t\right\} \right] = \diff y \left[ \frac {W^{(r)}(x)W^{(r)}(a-y)} {W^{(r)}(a)} -  W^{(r)} (x-y)\right], \quad y > 0.
\end{align*}
Hence
\begin{align*}
\int_0^\infty \diff t  \left[ e^{-rt} \p^x \left\{ X_{t-} \in \diff y, \sigma_A \wedge \tau_B > t\right\} \right] &= \int_0^\infty \diff t \left[ e^{-rt} \p^{x-A} \left\{ X_{t-} \in \diff (y-A), \sigma_0 \wedge \tau_{B-A} > t\right\} \right]  \\ &= \diff y \left[ \frac {W^{(r)}(x-A)W^{(r)}(B-y)} {W^{(r)}(B-A)} -  W^{(r)} (x-y)\right],
\end{align*}
when $y > A$, and it is zero otherwise. Therefore, for $u > A$, we have
\begin{multline*}
\int_0^\infty \diff t \left[ e^{ -r t} \p^x \{X_{t-} < u, \sigma_A \wedge \tau_B > t \} \right] = \int_A^{u} \diff y \left[ \frac {W^{(r)}(x-A)W^{(r)}(B-y)} {W^{(r)}(B-A)} -  W^{(r)} (x-y)\right] \\
\begin{aligned}
&= \int_0^{u-A} \diff z \left[ \frac {W^{(r)}(x-A)W^{(r)}(B-z-A)} {W^{(r)}(B-A)} -  W^{(r)} (x-z-A)\right] \\
&=  \frac {W^{(r)}(x-A)} {W^{(r)}(B-A)} \int_0^{u\wedge B-A} \diff z W^{(r)}(B-z-A) - \int_0^{u\wedge x-A} \diff z W^{(r)}(x-z-A)
\end{aligned}
\end{multline*}
since $W^{(r)}$ is zero on $(-\infty,0)$.  Therefore, $\E^x\left[  e^{-r (\sigma_A \wedge \tau_B)}1_{\{\sigma_A \wedge \tau_B = {\sigma_0} \}} \right]= \frac {W^{(r)}(x-A)} {W^{(r)}(B-A)}  \kappa(B;A) - \kappa(x;A)$.
Finally, substituting \eqref{about_lambda} in \eqref{definition_upsilon}, the proof is complete.
\end{proof}

\begin{proof}[Proof of Lemma \ref{convergence_kappa}]
(1) The monotonicity is clear because $\partial  \kappa(x;A) / {\partial A} = - W^{(r)}(x-A) \Pi(A,\infty) < 0$ for any $x > A > 0$.

(2) By \eqref{scale_function_asymptotic}, we have for any $u > A$
\begin{multline*}
\int_0^{u\wedge x-A} \diff z W^{(r)}(x-z-A) = \int_0^{u\wedge x-A} \diff z e^{\Phi(r)(x-z-A)}W_{\Phi(r)}(x-z-A) \\ \leq \frac 1 {\phi'(\lapinv)} \int_0^{u-A} \diff z e^{\Phi(r)(x-z-A)} = \frac {e^{\Phi(r)(x-A)}} {\Phi(r) \phi'(\lapinv)} \left( 1 - e^{-\Phi(r)(u-A)} \right).
\end{multline*}
Therefore,
\begin{align}
\kappa(x;A) \leq \frac {e^{\Phi(r)(x-A)}} {\Phi(r) \phi'(\lapinv)} \rho(A) \leq \frac {e^{\Phi(r)x}} {\Phi(r) \phi'(\lapinv)} \rho(0). \label{bound_kappa_by_rho}
\end{align}
Using this with the dominated convergence theorem yields the limit:
\begin{align*}
\kappa(x;0) = \lim_{A \downarrow 0}\frac 1 r \int_0^{\infty} \Pi(\diff u + A)  \left[ Z^{(r)}(x-A) - Z^{(r)}(x-A-u) \right] =  \frac 1 r \int_0^{\infty} \Pi(\diff u)  \left[ Z^{(r)}(x) - Z^{(r)}(x-u) \right],
\end{align*}
which is finite.

(3)
For all $x > A \geq 0$
\begin{align*}
\frac {\kappa(x;A)} {W^{(r)}(x-A)} = \int_A^{\infty} \Pi(\diff u) \int_0^{u\wedge x-A} \diff z \frac {W^{(r)}(x-z-A)} {W^{(r)}(x-A)} \leq \int_A^{\infty} \Pi(\diff u) \int_0^{u\wedge x-A}  e^{-\Phi(r) z} \diff z \leq \frac {\rho(A)} {\Phi(r)}.
\end{align*}
Therefore,  the dominated convergence theorem yields the limit:
\begin{align*}
\lim_{x \uparrow \infty}\frac {\kappa(x;A)} {W^{(r)}(x-A)} = \frac 1 r \int_A^{\infty} \Pi(\diff u)  \lim_{x \uparrow \infty} \frac { Z^{(r)}(x-A) - Z^{(r)}(x-u) } {W^{(r)}(x-A)} = \frac {\rho(A)} {\Phi(r)}
\end{align*}
where the last equality holds by \eqref{convergence_z_by_w}, $Z^{(r)}(x-A)/W^{(r)}(x-A) \xrightarrow{x \uparrow \infty} r/\Phi(r)$ and
\begin{align*}
\lim_{x \uparrow \infty}\frac {Z^{(r)}(x-u)} {W^{(r)}(x-A)} = \lim_{x \uparrow \infty} e^{- \Phi(r) (u-A) }\frac {Z^{(r)}(x-u)} {W^{(r)}(x-u)} \frac {W_{\Phi(r)}(x-u)} {W_{\Phi(r)}(x-A)} = e^{- \Phi(r) (u-A) } \frac r {\lapinv}.
\end{align*}
\end{proof}

\begin{proof}[Proof of Lemma \ref{remark_asymptotics_Psi}](1) It is immediate by Lemma \ref{convergence_kappa}-(3) and \eqref{convergence_z_by_w}.
(2) By Lemma \ref{convergence_kappa}-(2) and because $\rho(A) \xrightarrow{A \downarrow 0}\rho(0)$, the convergence indeed holds.
(3) By \eqref{bound_kappa_by_rho}, the dominated convergence theorem yields
\begin{align*}
\lim_{B \downarrow A} \Psi(A,B) = \lim_{B \downarrow A}\Big[ \Big( \frac \pcheck r - \costb \Big) - \Big( \frac \pcheck r +\costs \Big) Z^{(r)} (B-A) + \left( \acheck - \costs \right) \kappa(B;A) \Big] = - (\gamma_b + \gamma_s) < 0.
\end{align*}
\end{proof}

\begin{proof}[Proof of Remark \ref{lemma_creeping}]
By Theorem 8.1 of \cite{Kyprianou_2006}, we obtain the limits:
\begin{align*}
\lim_{A \downarrow 0} \E^x\left[  e^{-r (\sigma_A \wedge \tau_B)} 1_{\{\tau_B < \sigma_A, \tau_B \wedge \sigma_A < \infty \}} \right]  &= \E^x \left[ e^{-r \tau_B} 1_{\{\tau_B < {\sigma_0}, \, \tau_B < \infty \}}\right], \\ \lim_{A \downarrow 0}\E^x\left[  e^{-r (\sigma_A \wedge \tau_B)} 1_{\{\tau_B > \sigma_A \, \textrm{or} \;  \sigma_A \wedge \tau_B = {\sigma_0} \}} 1_{\{ \tau_B \wedge \sigma_A < \infty \}} \right] &=\E^x\left[  e^{-r \tau_B} 1_{\{\tau_B = {\sigma_0} < \infty \}}\right].
\end{align*}
By the construction of $\E^x\left[  e^{-r (\sigma_A \wedge \tau_B)}1_{\{\sigma_A \wedge \tau_B = {\sigma_0} < \infty \}} \right]$ as seen in \eqref{lambda_3_construction} above, we deduce that
\begin{align*}
\lim_{A \downarrow 0}\E^x\left[  e^{-r (\sigma_A \wedge \tau_B)}1_{\{\sigma_A \wedge \tau_B = {\sigma_0} < \infty \}} \right] &=\E^x\left[  e^{-r \tau_B} 1_{\{X_{\tau_B} < 0, \tau_B < \infty \}}\right] \\&= \E^x\left[  e^{-r \tau_B} 1_{\{\tau_B = {\sigma_0} < \infty \}}\right] - \E^x\left[  e^{-r \tau_B} 1_{\{X_{\tau_B} = 0, \, \tau_B < \infty \}}\right].
\end{align*}
Applying these to the definition \eqref{definition_upsilon} yields:
\begin{align*}
\Upsilon(x;0+,B) = \Upsilon(x;0,B) - \left( \acheck - \gamma_s \right) \E^x \left[ e^{-r \tau_B}  1_{\{X_{\tau_B} = 0, \, \tau_B < \infty \}} \right].
\end{align*}
By \cite{Kyprianou_2006} Exercise 7.6, a spectrally negative \lev process creeps downward, or $\p \left\{ X_{\sigma_0} = 0 \, | {\sigma_0} < \infty \right\} > 0$, if and only if there is a Gaussian component.  This completes the proof.
\end{proof}

\begin{proof}[Proof of Lemma \ref{lemma_gamma_B}]  We first show the following.
\begin{lemma} \label{lemma_bounded_psi_zero}
If $\int_0^1 u \Pi (\diff u) < \infty$, then we have  $\int_0^{\infty} \Pi (\diff u) \big( 1 - \frac {W^{(r)}(B-u)} {W^{(r)}(B)} \big) < \infty$  for any $0 < B < \infty$.
\end{lemma}
\begin{proof}
Fix $B > 0$. We have
\begin{align}
\int_0^\infty \Pi(\diff u) \Big( 1 - \frac {W^{(r)} (B-u)} {W^{(r)} (B)}\Big) = \Pi(B,\infty) + \frac 1 {W^{(r)}(B)} \int_0^B \Pi(\diff u) \left( W^{(r)}(B) - W^{(r)} (B-u)\right). \label{int_scale_function_ratio}
\end{align}
For any $0 < \epsilon < B$, we have by the mean value theorem,
\begin{align*}
\int_0^\epsilon (W^{(r)}(B) - W^{(r)} (B-u)) \Pi (\diff u) \leq \int_0^\epsilon u \sup_{t \in [B-\epsilon, B]} W^{(r)'}(t) \Pi (\diff u)
\end{align*}
which is finite because $\sup_{t \in [B-\epsilon, B]} W^{(r)'}(t) < \infty$ and $\int_0^1 u \Pi (\diff u) < \infty$.  Hence we conclude.\end{proof}

(1) Suppose $B < \infty$. Since $W^{(r)}(B-u) / W^{(r)}(B-A)$ is increasing in $A$ on $(0,B)$, it follows that
\begin{align*}
\frac \partial {\partial A}\widehat{\psi}(A,B) =   -\left( \acheck - \costs \right)\int_{A}^{B}  \Pi (\diff u) \frac \partial {\partial A} \left[ \frac {W^{(r)}(B-u)} {W^{(r)}(B-A)} \right] < 0, \quad 0 < A < B,
\end{align*}
and $\widehat{\psi}$ is decreasing in $A$ on $(0,B)$.  The result for $B = \infty$ is immediate because $\rho(A)$ is decreasing.

For the convergence result for $B < \infty$ (when $\int_0^1 u \Pi (\diff u) < \infty$), we have
\begin{align*}
\int_A^\infty \Pi(\diff u) \left[ 1 - \frac {W^{(r)}(B-u)} {W^{(r)}(B-A)} \right] \leq \frac 1 {W^{(r)}(B-A)} \int_0^\infty \Pi (\diff u) \left[ W^{(r)}(B) - W^{(r)}(B-u)\right],
\end{align*}
which is bounded by Lemma \ref{lemma_bounded_psi_zero}.  Hence by the dominated convergence theorem,
\begin{align*}
\lim_{A \downarrow 0}\int_A^\infty \Pi(\diff u) \left[ 1 - \frac {W^{(r)}(B-u)} {W^{(r)}(B-A)} \right] &= \int_0^\infty \lim_{A \downarrow 0} \Pi(\diff u+A) \left[ 1 - \frac {W^{(r)}(B-u-A)} {W^{(r)}(B-A)} \right] \\ &= \int_0^\infty \Pi(\diff u) \left[ 1 - \frac {W^{(r)}(B-u)} {W^{(r)}(B)} \right].
\end{align*}
The convergence result for $B = \infty$ is clear because $\rho(A) \xrightarrow{A \downarrow 0} \rho(0)$.

(2) Suppose $A > 0$.  Look at \eqref{gamma_b_definition} and consider the derivative with respect to $B$,
\begin{align*}
\frac \partial {\partial B} \widehat{\psi}(A,B) =- \left( \acheck - \costs \right) \left[ \pi(B) \frac {W^{(r)}(0)} {W^{(r)}(B-A)} + \int_{A}^B \Pi (\diff u) \frac \partial {\partial B} \frac {W^{(r)}(B-u)} {W^{(r)}(B-A)} \right]
\end{align*}
where $\pi$ is the density of $\Pi$.  Moreover, for all $A < u < B$,
\begin{align*}
 \frac \partial {\partial B} \frac {W^{(r)}(B-u)} {W^{(r)}(B-A)}& =  e^{-\Phi(r) (u-A)}\frac \partial {\partial B} \frac {W_{\Phi(r)}(B-u)} {W_{\Phi(r)}(B-A)} \\ &= e^{-\Phi(r) (u-A)} \frac {W_{\Phi(r)}'(B-u) W_{\Phi(r)}(B-A) - W_{\Phi(r)}(B-u) W_{\Phi(r)}'(B-A)} {(W_{\Phi(r)}(B-A))^2} \ge 0,
\end{align*}
 by \eqref{assumeW}. Therefore, $\widehat{\psi}(A,B)$ is decreasing in $B$.  This result can be extended to $A=0$ as in part (1).

For the convergence result for $A > 0$, the dominated convergence theorem yields
\begin{align*}
\lim_{B \rightarrow \infty}\int_A^\infty \Pi (\diff u) \left( 1 - \frac {W^{(r)}(B-u)} {W^{(r)}(B-A)} \right) = \int_A^\infty \Pi (\diff u)  \lim_{B \rightarrow \infty} \left( 1 - \frac {W^{(r)}(B-u)} {W^{(r)}(B-A)} \right)    = \rho(A),
\end{align*}
where the last equality holds by \eqref{w_phi}-\eqref{scale_function_asymptotic}.

When $A = 0$, it also holds by applying the dominated convergence theorem.  Indeed, \eqref{int_scale_function_ratio} is bounded in $B$ on $[B_0, \infty)$ for any $B_0 > 0$.  To see this, for any $0 < \varepsilon < B$
\begin{align*}
&\frac 1 {W^{(r)}(B)} \int_0^B \Pi(\diff u) \left( W^{(r)}(B) - W^{(r)} (B-u)\right)  \\ &=
\frac {e^{\Phi(r) B}} {W^{(r)}(B)} \left( \int_0^B \Pi (\diff u) W_{\Phi(r)}(B) \left[ 1 - e^{-\Phi(r)u}\right] + \int_0^B \Pi (\diff u) e^{-\Phi(r)u} \left[ W_{\Phi(r)}(B) - W_{\Phi(r)}(B-u)\right] \right) \\
&\leq \frac {e^{\Phi(r) B}} {W^{(r)}(B)} \left( W_{\Phi(r)}(B) \rho(0) + W_{\Phi(r)}(B) \Pi(\varepsilon, B) + \alpha(B; \varepsilon) \right), 
\end{align*}
with $\alpha(B; \varepsilon) :=\int_0^\varepsilon \Pi (\diff u) \left[ W_{\Phi(r)}(B) - W_{\Phi(r)}(B-u)\right]$.  Moreover for any $B > B_0 > \varepsilon$, by the mean value theorem, $\alpha(B; \varepsilon) \leq \int_0^\varepsilon u \sup_{t \geq B_0 - \varepsilon} W_{\Phi(r)}'(t) \Pi (\diff u)$
which is finite because $\sup_{t \geq B_0 - \varepsilon} W_{\Phi(r)}'(t) < \infty$ and $\int_0^1 u \Pi (\diff u) < \infty$.  This together with $W^{(r)}(x) \sim e^{\Phi(r)x}/\phi'(\Phi(r))$ as $x \uparrow \infty$ shows that \eqref{int_scale_function_ratio} is bounded in $B$ on $[B_0, \infty)$.


(3)  The derivative of \eqref{kappa_zero} can go into the integral by the dominated convergence theorem because $\frac 1 r \int_0^\infty \Pi(\diff u) \left| Z^{(r)'}(B) - Z^{(r)'}(B-u)\right| =  \int_0^\infty \Pi(\diff u) \left( W^{(r)}(B) - W^{(r)}(B-u)\right) < \infty$ by Lemma \ref{lemma_bounded_psi_zero}.  Therefore, the result follows.
\end{proof}

%

\begin{proof}[Proof of Theorem \ref{lemma_classification_b}]
(1) In view of (a)-(c) in Subsection \ref{subsection_cont_smooth_fit}, we shall show that (i) $\Psi(\underline{A},B)$ monotonically increases while (ii) $\Psi(\overline{A},B)$ monotonically decreases in $B$.

(i) By the assumption $\underline{A} > 0$, we have $\widehat{\psi}(\underline{A},\infty) = 0$.  This coupled with the fact that $\widehat{\psi}(\underline{A},B)$ is decreasing in $B$ by Lemma \ref{lemma_gamma_B}-(2) shows that $\widehat{\psi}(\underline{A},B) > 0$ or $\psi(\underline{A},B) > 0$ for every $B > \underline{A}$ and hence $\Psi(\underline{A},B)$ is monotonically increasing in $B$ on $(\underline{A},\infty)$ (recall $\psi(\underline{A},B) = \partial  \Psi (\underline{A},B) / {\partial B}$).  Furthermore, $\underline{b}(\underline{A}) < \infty$ implies that $\widehat{\Psi}(\underline{A},\infty) > 0$ (note $\widehat{\Psi}(\underline{A},B) > 0 \Longleftrightarrow \Psi(\underline{A},B) > 0$). This together with $W^{(r)}(B-\underline{A}) \xrightarrow{B \uparrow \infty} \infty$ implies that $\Psi(\underline{A},B)$ is monotonically increasing in $B$ to $+\infty$.

(ii) Because $\overline{A} \geq \underline{A}$, we obtain $\overline{A} > 0$ and hence $\widehat{\psi}(\overline{A},\overline{A}+) = 0$. This together with the fact that $\widehat{\psi}(\overline{A},B)$ is decreasing in $B$ by Lemma \ref{lemma_gamma_B}-(2) shows that  $\widehat{\psi}(\overline{A},B) < 0$, or $\psi(\overline{A},B) < 0$, for every $B > \overline{A}$.  Consequently, $\Psi(\overline{A},B)$ is monotonically decreasing in $B$ on $(\overline{A},\infty)$. Furthermore, because $\Psi(\overline{A},\overline{A}+) < 0$ by Lemma \ref{remark_asymptotics_Psi}-(3), $\Psi(\overline{A},B)$ never up-crosses the level zero.

By (i) and (ii) and the continuity of $\Psi$ and $\psi$ with respect to both $A$ and $B$, there must exist $A^* \in (\underline{A},\overline{A})$ and $B^* \in (A^*,\infty)$ such that $B^* = \underline{b}(A^*) = \overline{b}(A^*)$ (with $\Psi(A^*,B^*)=\psi(A^*,B^*)=0$).

(2) Using the same argument as in (1)-(i) above, $\Psi(\underline{A},B)$ is increasing in $B$ on $(\underline{A},\infty)$.  Moreover, the assumption $\underline{b}(\underline{A}) = \infty$ means that $-\infty < \Psi(\underline{A},\underline{A}+) \leq \lim_{B \uparrow \infty}\Psi(\underline{A},B) \leq 0$. This together with $W^{(r)}(B-A) \xrightarrow{B \uparrow \infty} \infty$ shows $\widehat{\Psi}(\underline{A},\infty)=0$. By \eqref{asymptotics_ratio_with_derivative} and \eqref{eq_behavior_gamma}, $\widehat{\psi}(\underline{A},\infty)=0$ and this implies that $\widehat{\psi}(\underline{A},B)> 0$ for all $B \in (\underline{A},\infty)$ by virtue of Lemma \ref{lemma_gamma_B}-(2), and hence $\overline{b}(\underline{A}) = \infty$.  

(3) Recall Lemma \ref{lemma_gamma_B}-(3).  We have $\psi(0,B) > 0$ if and only if $B \in (0,\overline{b}(0))$, and hence $\Psi(0,\cdot)$ attains a global maximum $\Psi(0,\overline{b}(0))$ and it is strictly larger than zero because $\underline{b}(0) < \overline{b}(0)$.  Furthermore, $\Psi(\overline{A},B)$ is monotonically decreasing in $B$ on $(\overline{A},\infty)$ and $\Psi(\overline{A},\overline{A}+) < 0$ as in (1)-(ii).  This together with the same argument as in (1) shows the result.

(4) First, $\overline{A} = 0$ implies $\overline{b}(0) = 0$.  This also means that $\widehat{\psi}(0,B) \leq 0$ or $\Psi(0,B)$ is decreasing on $(0,\infty)$.  This together with Lemma \ref{remark_asymptotics_Psi}-(3) shows $\underline{b}(0) = \infty$.  Now, for both (i) and (ii) for every $B \in [\overline{b}(0),\underline{b}(0)]$, because $\psi(0,B) \leq 0$, we must have $\widehat{\Psi}(0,B) - \widehat{\psi}(0,B) \frac {W^{(r)}(B)} {W^{(r)'}(B)} \geq \widehat{\Psi}(0,B)$.  This shows that $b(0) \leq \underline{b}(0)$.  It is clear that this is \textbf{case 3} when $b(0) < \infty$ whereas this is \textbf{case 4} when $b(0) = \infty$.
\end{proof}

\begin{proof}[Proof of Lemma \ref{lemma_b_bar}]
(1) With $W^{(r)}(B-A) > 0$, it is sufficient to show  $\Psi(A,B)$ is decreasing in $A$ on $(\underline{A},\overline{A})$ for every fixed $B$.  Indeed, the derivative
\begin{align}
\frac \partial {\partial A}\Psi(A,B)  &= \frac \partial {\partial A} \Big[ - \Big( \frac \pcheck r + \gamma_s \Big) Z^{(r)}(B-A) + (\acheck - \gamma_s) \kappa(B;A) \Big] \\ &= W^{(r)} (B-A) \left( \pcheck + r\costs - (\acheck - \gamma_s)\Pi(A,\infty) \right) \label{Psi_derivative_A}
\end{align}
 is negative for every $A \in (0,\overline{A})$ by the definition of $\overline{A}$.
Part (2) is immediate from  Lemma \ref{lemma_gamma_B}-(1).
\end{proof}

\begin{proof}[Proof of Lemma \ref{lemma_domination}]
(1) Fix $B^* > x > A > A^* > 0$. First, suppose $B^* < \infty$. We compute the derivative:
\begin{multline*}
\frac \partial {\partial A} (v_{A,B^*}(x) - g(x)) = \frac \partial {\partial A} \Upsilon(x;A,B^*)   = \Big[ \frac \partial {\partial A} \frac {W^{(r)}(x-A)} {W^{(r)}(B^*-A)} \Big] \Psi(A,B^*) \\ +   \frac {W^{(r)}(x-A)} {W^{(r)}(B^*-A)}  \frac \partial {\partial A} \Psi(A,B^*)  + \frac \partial {\partial A} \Big[ \Big( \frac \pcheck r + \costs \Big) Z^{(r)}(x-A) - \left( \acheck - \costs \right) \kappa(x;A) \Big].
\end{multline*}
Using \eqref{Psi_derivative_A}, the last two terms of the above cancel out and
\begin{align*}
\frac \partial {\partial A} (v_{A,B^*}(x) - g(x))   = \Big[ \frac \partial {\partial A} \frac {W^{(r)}(x-A)} {W^{(r)}(B^*-A)} \Big] \Psi(A,B^*).
\end{align*}
On the right-hand side, the derivative is given by
\begin{align*}
\frac \partial {\partial A} \frac {W^{(r)}(x-A)} {W^{(r)}(B^*-A)} &= e^{-\Phi(r)(B^*-x)}\frac \partial {\partial A} \frac {W_{\Phi(r)}(x-A)} {W_{\Phi(r)}(B^*-A)}  \\ &= e^{-\Phi(r)(B^*-x)} \frac {-W_{\Phi(r)}'(x-A) W_{\Phi(r)}(B^*-A) + W_{\Phi(r)}(x-A) W_{\Phi(r)}'(B^*-A)} {W_{\Phi(r)}(B^*-A)^2}
\end{align*}
which is negative  according to \eqref{assumeW} by $B^* > x$.
Now suppose $B^* = \infty$. We have
\begin{multline*}
\frac \partial {\partial A} (v_{A,\infty}(x) - g(x))  =  \frac \partial {\partial A} \left[ W^{(r)}(x-A)\widehat{\Psi}(A,\infty) \right]  + \frac \partial {\partial A} \Big[ \Big( \frac \pcheck r + \costs \Big) Z^{(r)}(x-A) - \left( \acheck - \costs \right) \kappa(x;A) \Big].
\end{multline*}
By \eqref{gamma_a_definition}, the first term becomes
\begin{align*}
\frac \partial {\partial A} \left[ W^{(r)}(x-A)\widehat{\Psi}(A,\infty) \right] &= - W^{(r)'}(x-A) \widehat{\Psi}(A,\infty) - (\widetilde{\alpha} - \gamma_s) W^{(r)}(x-A) \int_A^\infty \Pi(\diff u) e^{-\Phi(r) (u-A)},
\end{align*}
and by using the last equality of \eqref{Psi_derivative_A} (with $B$ replaced with $x$), we obtain
\begin{multline*}
- (\widetilde{\alpha} - \gamma_s) W^{(r)}(x-A) \int_A^\infty \Pi(\diff u) e^{-\Phi(r) (u-A)} + \frac \partial {\partial A} \Big[ \Big( \frac \pcheck r + \costs \Big) Z^{(r)}(x-A) - \left( \acheck - \costs \right) \kappa(x;A) \Big] \\
= W^{(r)} (x-A) \left( - (\pcheck + r\costs) + (\acheck - \gamma_s) \rho(A) \right) = W^{(r)}(x-A) \Phi(r) \widehat{\Psi}(A,\infty).
\end{multline*}
Hence,
\begin{align*}
\frac \partial {\partial A} (v_{A,\infty}(x) - g(x)) = - \left[ W^{(r)'}(x-A)  - \Phi(r) W^{(r)}(x-A) \right] \widehat{\Psi}(A,\infty) = - e^{\Phi(r)(x-A)} W_{\Phi(r)}'(x-A) \widehat{\Psi}(A,\infty)
\end{align*}
where  $W_{\Phi(r)}'(x-A) > 0$ because $W_{\Phi(r)}$ is increasing.

Now in order to show  $v_{A,B^*}(x) - g(x)$ is increasing in $A$ on $(A^*,x)$, it is sufficient to show $\widehat{\Psi}(A,B^*) \leq 0$ for every $A^* < A < B^*$. This is true for $A^* < A < \overline{A}$ by $\underline{b}(A^*)=B^*$ and Lemma \ref{lemma_b_bar}-(1).  This holds also for $\overline{A} \leq A < B^*$.  Indeed, $\Psi(A,B)$ is decreasing in $B$ because,  for any $B > A > \overline{A}$, $\widehat{\psi}(A,A+) < 0$ and Lemma \ref{lemma_gamma_B}-(2) imply $\psi(A, B) \leq 0$.  Furthermore,  Lemma \ref{remark_asymptotics_Psi}-(3) shows that $\Psi(A,A+) < 0$.  Hence $\Psi(A,B^*) \leq 0$ or $\widehat{\Psi}(A,B^*) \leq 0$.

Now we have by \eqref{eq_continuous_fit_A}, $0 \geq {W^{(r)}(0)} \widehat{\Psi}(x,B^*) = v_{x,B^*}(x+) - g(x) \geq  v_{A^*,B^*}(x) - g(x)$.
This proves \eqref{ineq_g} for the case $A^*>0$.  Since $v_{0+,B^*}(x) = \lim_{A \downarrow 0} v_{A,B^*}(x)$ by \eqref{delta_by_upsilon} and \eqref{upsilon_zero}, this also shows for the case $A^* = 0$.

(2) Recall that $\psi(A^*,B) = \partial \Psi(A^*,B) / \partial B$ and hence for any $A^* < x < B < B^*$
\begin{align*}
\frac \partial {\partial B} (v_{A^*,B}(x) - h(x)) &= \frac \partial {\partial B} \Upsilon(x;A^*,B) \\ &= \frac {W^{(r)}(x-A^*)} {(W^{(r)}(B-A^*))^2} \big[ \psi(A^*,B)  W^{(r)}(B-A^*) - \Psi(A^*,B) W^{(r)'}(B-A^*)\big] \\ &= -  W^{(r)}(x-A^*)\frac {W^{(r)'}(B-A^*)} {W^{(r)}(B-A^*)}  \Big(\widehat{\Psi}(A^*,B) - \widehat{\psi}(A^*,B) \frac {W^{(r)}(B-A^*)} {W^{(r)'}(B-A^*)} \Big)
\end{align*}
which is positive for $B \in (A^*,B^*)$ by  Remark \ref{remark_A_B_sign}-(1). Therefore, by \eqref{eq_continuous_fit_B},
$0 = v_{A^*,x}(x-) - h(x)  \leq  v_{A^*,B^*}(x) - h(x)$.

This proves \eqref{ineq_h} for the case $B^* < \infty$. Since $v_{A^*, \infty}(x)= \lim_{B \uparrow \infty}v_{A^*, B}(x)$ by \eqref{delta_by_upsilon} and \eqref{upsilon_infty}, this also shows for the case $B^* = \infty$.
\end{proof}

\begin{proof}[Proof of Lemma \ref{lemma_domination2}]
(1) Suppose $A^* > 0$.  Because $X_{\sigma_{A^*} \wedge \tau} > A^*$ a.s.\ on $\{\tau < \sigma_{A^*}, \tau < \infty \}$,  $X_{\sigma_{A^*} \wedge \tau} \leq A^*$ a.s.\ on $\{\tau \geq \sigma_{A^*}, \sigma_{A^*} < \infty \}$ and by \eqref{ineq_h},  we have on $\{ \tau \wedge \sigma_{A^*} < \infty\}$
\begin{multline*}
 g(X_{\sigma_{A^*}}) 1_{\{\sigma_{A^*} < \tau \}} + h(X_{\tau}) 1_{\{\tau < \sigma_{A^*}\}}
\leq  g(X_{\sigma_{A^*}}) 1_{\{\sigma_{A^*} < \tau \}} + v_{A^*,B^*}(X_{\tau}) 1_{\{\tau < \sigma_{A^*}\}} \\
=  v_{A^*,B^*}(X_{\sigma_{A^*}}) 1_{\{\sigma_{A^*} < \tau \}} + v_{A^*,B^*}(X_{\tau}) 1_{\{\tau < \sigma_{A^*}\}} =  v_{A^*,B^*}(X_{\sigma_{A^*} \wedge \tau}).
\end{multline*}
Suppose $A^* = 0$.  We have, by \eqref{ineq_h}, on $\{ \tau < \infty \}$
\begin{align*}
- (\acheck - \gamma_s) 1_{\{X_{\tau} = 0\}} + h(X_{\tau}) 1_{\{\tau < {\sigma_0} \}}
\leq   - (\acheck - \gamma_s) 1_{\{X_{\tau} = 0\}} + v_{0+,B^*}(X_{\tau}) 1_{\{\tau < {\sigma_0} \}}  =  v_{0+,B^*}( X_\tau).
\end{align*}
The proof for (2) is similar thanks to \eqref{ineq_g}.
\end{proof}

\begin{proof}[Proof of Lemma \ref{lemma_generator}]
(1) First, Lemma 3.4 of \cite{Leung_Yamazaki_2010} shows that $(\mathcal{L}-r)\zeta(x) = 0$.
Therefore,  using  \eqref{value_function} and that $J'=J''=0$ on $(0,A^*)$, we have
\begin{align}
(\mathcal{L}-r) v_{A^*,B^*}(x) = \int_x^\infty \left( J(x-u) - J(x) \right) \Pi(\diff u)  - r J(x) =  \left( \acheck - \costs \right) \Pi(x,\infty)  - (r \costs + \pcheck). \label{generator_case_1}
\end{align}
Since $A^* > 0$, we must have by construction $\widehat{\Psi}(A^*,B^*) = 0$ and $\widehat{\Psi}(A^*,B^*) - \widehat{\psi}(A^*,B^*) \frac {W^{(r)}(B^*-A^*)} {W^{(r)'}(B^*-A^*)} = 0$ and consequently, $\widehat{\psi}(A^*,B^*) = 0$.  Furthermore, $\widehat{\psi}(A^*,B)$ is decreasing in $B$ and  hence $\widehat{\psi}(A^*,A^*+) =  \left( \acheck - \costs \right) \Pi(A^*,\infty) -\left( \pcheck + \costs r \right) > 0$.  Applying this to  \eqref{generator_case_1}, for $x < A^*$, it follows that $(\mathcal{L}-r) v_{A^*,B^*}(x) > 0$.

(2)
When $A^* > 0$, by the strong Markov property,
\begin{multline*}
e^{-r (t \wedge \sigma_{A^*} \wedge \tau_{B^*})}v_{A^*,B^*}(X_{t \wedge \sigma_{A^*} \wedge \tau_{B^*}}) \\ =  \E^x \left[ \left. e^{-r (\tau_{B^*} \wedge \sigma_{A^*})} \left( h(X_{\tau_{B^*}}) 1_{\{\tau_{B^*} < \sigma_{A^*}\}} + g(X_{\sigma_{A^*}}) 1_{\{\tau_{B^*} > \sigma_{A^*} \}}\right) 1_{\{\tau_{B^*} \wedge \sigma_{A^*} < \infty \}} \right| \mathcal{F}_{t \wedge \sigma_{A^*} \wedge \tau_{B^*}}\right].
\end{multline*}
Taking expectation on both sides, we see that $e^{-r (t \wedge \sigma_{A^*} \wedge \tau_{B^*})}v_{A^*,B^*} (X_{t \wedge \sigma_{A^*} \wedge \tau_{B^*}}) $ is a $\p^x$-martingale and hence $(\mathcal{L}-r)v_{A^*,B^*}(x) = 0$ on $(A^*,B^*)$ (see Remark \ref{remark_smoothness} and the Appendix of \cite{Biffis_Kyprianou_2010}).

When $A^* = 0$ by Remark \ref{lemma_creeping}
\begin{align*}
e^{-r (t \wedge \tau_{B^*})}v_{0+,B^*}(X_{t  \wedge \tau_{B^*}})  =  \E^x \big[  e^{-r \tau_{B^*}} \big( h(X_{\tau_{B^*}}) 1_{\{\tau_{B^*} < {\sigma_0} \}} - (\acheck - \gamma_s) 1_{\{X_{\tau_{B^*}} = 0 \}}\big) 1_{\{\tau_{B^*} < \infty \}} \big| \mathcal{F}_{t \wedge \tau_{B^*}}\big].
\end{align*}
Taking expectation on both sides, we see that $e^{-r (t \wedge \tau_{B^*})}v_{0+,B^*}(X_{t  \wedge \tau_{B^*}}) $ is a $\p^x$-martingale and hence $(\mathcal{L}-r)v_{0+,B^*}(x) = 0$ on $(0,B^*)$.


(3)
Suppose $\nu > 0$, i.e.\ there is a Gaussian component.   In this case,  $W^{(r)}$ is continuous on $\R$ and $C^2$ on $(0,\infty)$, and we have
\begin{multline*}
v_{A^*,B^*}''(B^*-) - h''(B^*) = W^{(r)''}(B^*-A^*) \widehat{\Psi}(A^*,B^*) \\ + (\pcheck + \gamma_s r) W^{(r)'}(B^*-A^*) - (\acheck - \gamma_s) \int_{A^*}^\infty \Pi(\diff u) \big( W^{(r)'}(B^*-A^*)  - W^{(r)'}(B^*-u) \big).
\end{multline*}
We show $v_{A^*,B^*}''(B^*-) - h''(B^*) \geq 0$. To this end, we suppose $v_{A^*,B^*}''(B^*-) - h''(B^*)  < 0$ and derive contradiction. The fact that $v_{A^*,B^*}'(B^*-) - h'(B^*) =0$ by smooth fit implies that $v_{A^*,B^*}'(x) - h'(x) > 0$ for some $x\in (B^*-\varepsilon, B^*)$.  However, since $v_{A^*,B^*}(B^*-) - h(B^*)  =0$, this would contradict \eqref{ineq_h}.  Consequently, $v_{A^*,B^*}''(B^*-) - h''(B^*)  \geq 0$, implying $(\mathcal{L}-r)v_{A^*,B^*}(B^*+) \leq (\mathcal{L}-r)v_{A^*,B^*}(B^*-)$.  When $\nu=0$, $(\mathcal{L}-r)v_{A^*,B^*}(B^*+) = (\mathcal{L}-r)v_{A^*,B^*}(B^*-)$ by continuous and smooth fit.

As a result, for all cases, we conclude that $(\mathcal{L}-r)v_{A^*,B^*}(B^*+) \leq (\mathcal{L}-r)v_{A^*,B^*}(B^*-) = 0$.
Now it is sufficient to show that $(\mathcal{L}-r)v_{A^*,B^*}(x) $ is decreasing on $(B^*, \infty)$.  Recall the decomposition \eqref{value_function}. Because $(\mathcal{L}-r)\zeta(x)=0$, we shall show $(\mathcal{L}-r)J(x) $  is decreasing on $(A^*,B^*)$.

 Now because $J'=J''=0$ on $x > B^*$,
\begin{align*}
(\mathcal{L}-r)J(x) = \int_{x-B^*}^\infty \Pi(\diff u) \left[ J(x-u) - \Big(\frac p r - \costb \Big) \right] - (p - r \costb), \quad x > B^*.
\end{align*}
Since $v_{A^*,B^*}(x) \geq h(x)$, we must have that
$J(x) \geq \frac p r - \costb$ on $x < B^*$ (or the integrand of the above is non-negative).  In order to show that this is decreasing, we show that $J$ in \eqref{def_J} is decreasing on $(-\infty, B^*)$.  By continuous fit at $A^*$ (when $A^* > 0$), it is sufficient to show that $\Upsilon(x; A^*, B^*)$ is \emph{decreasing} for every $x \in (A^*, B^*)$.  By Remark \ref{remark_A_B_sign}-(3), we must have $\widehat{\Psi}(A^*,B^*) - \widehat{\psi}(A^*,B^*) \frac {W^{(r)}(B^*-A^*)} {W^{(r)'}(B^*-A^*)}=0$, and hence by \eqref{assumeW} and because $\Psi(A^*,B^*) \leq 0$ as in  Remark \ref{remark_A_B_sign}-(1),
\begin{align*}
0 = \frac {W^{(r)'}(B^*-A^*)} {W^{(r)}(B^*-A^*)} \Psi(A^*,B^*) - \psi(A^*, B^*) \geq \frac {W^{(r)'}(x-A^*)} {W^{(r)}(x-A^*)} \Psi(A^*,B^*) - \psi(A^*, B^*).
\end{align*}
After multiplying by $W^{(r)}(x-A^*)/W^{(r)}(B^*-A^*)$ on both sides and observing $\widehat{\psi}(A^*, x)$ is decreasing in $x$ by Lemma \ref{lemma_gamma_B}, we get $0  \geq W^{(r)'}(x-A^*) \widehat{\Psi}(A^*,B^*) - W^{(r)}(x-A^*)\widehat{\psi}(A^*, B^*)  \geq W^{(r)'}(x-A^*) \widehat{\Psi}(A^*,B^*) - W^{(r)}(x-A^*)\widehat{\psi}(A^*, x)$,
which  matches $\Upsilon'(x;A^*,B^*)$ in \eqref{deltaprimes}.  Hence, $\Upsilon(x; A^*, B^*)$ is decreasing, as desired.\end{proof}

\begin{proof}[Proof of Theorem \ref{theorem_equilibrium}]
(i)  We show that $v_{A^*,B^*}(x) \geq v(x;\sigma_{A^*}, \tau)$ for every $\tau \in \S$.  As is discussed in Remark \ref{remark_laststep}, we only need to focus on the set $\S_{A^*}$.

In order to handle the discontinuity of $v_{A^*,B^*}$ at zero, we first construct a sequence of functions $v_n(\cdot)$ such that it is continuous on $\R$, $v_n(x) = v_{A^*,B^*}(x)$ on $x \in (0,\infty)$ and (c) $v_n(x) \uparrow v_{A^*,B^*}(x)$ pointwise for every fixed $x \in (-\infty,0)$.
 Notice that $v_{A^*,B^*}(\cdot)$ is uniformly bounded  because $h(\cdot)$ and $g(\cdot)$ are.  Hence, we can choose so that $v_n$ is also uniformly bounded for every fixed $n \geq 1$.  Because $v_{A^*, B^*}'(x)=v'_n(x)$ and $v_{A^*,B^*}''(x)=v_n''(x)$  on $x \in (0,\infty)\backslash \{A^*,B^*\}$ and $v_{A^*,B^*}(x) \geq v_n(x)$ on $(-\infty,0)$, we have
\begin{align}
(\mathcal{L}-r) (v_n - v_{A^*,B^*}) (x) \leq 0, \quad x \in (0,\infty) \backslash\{A^*,B^*\}. \label{generator_difference_negative}
\end{align}
We have for any $\tau \in \S_{A^*}$,
$\E^x \left[ \int_0^{\tau \wedge \sigma_{A^*}} e^{-rs} |(\mathcal{L}-r) (v_n - v_{A^*,B^*}) (X_{s-})| \diff s\right]  \leq K  \E^x \left[ \int_0^{\sigma_{A^*}} e^{-rs} \Pi(X_{s-},\infty) \diff s\right]$
where $K := \sup_{x \in \R}|v_{A^*,B^*}(x) - v_n (x)| < \infty$ is the maximum difference between $v_{A^*,B^*}$ and $v_n$.   Using $N$ as the Poisson random measure for the jumps of $-X$ and $\underline{X}$ as the running minimum of $X$, by the compensation formula \cite[Theorem 4.4]{Kyprianou_2006},
\begin{multline*}
\E^x \left[ \int_0^{\sigma_{A^*}} e^{-rs} \Pi(X_{s-},\infty) \diff s\right]  = \E^x \left[ \int_0^\infty \int_0^\infty e^{-rs} 1_{\{\underline{X}_{s-} > A^*, \; u > X_{s-}\}} \Pi(\diff u) \diff s \right] \\ = \E^x \left[ \int_0^\infty \int_0^\infty e^{-rs} 1_{\{\underline{X}_{s-} > A^*, \; u > X_{s-}\}} N (\diff u \times \diff s) \right]  = \E^x \left[ e^{-r \sigma_{A^*}} 1_{\{X_{\sigma_{A^*}} < 0, \, \sigma_{A^*} < \infty \}}\right] < \infty.
\end{multline*}
Therefore,  uniformly for any $n \geq 1$,
\begin{align}
\begin{split}
\E^x \left[ \int_0^{\tau \wedge \sigma_{A^*}} e^{-rs} |(\mathcal{L}-r) (v_n - v_{A^*,B^*}) (X_{s-})| \diff s\right] &< \infty, \\
\int_0^{\tau \wedge \sigma_{A^*}} e^{-rs} |(\mathcal{L}-r) (v_n - v_{A^*,B^*}) (X_{s-}) | \diff s &< \infty, \quad \p^x\textrm{-a.s.}
\end{split} \label{expectation_bounded_generator}
\end{align}


By applying Ito's formula to
$\left\{ e^{-r {(t \wedge \sigma_{A^*})}} v_{n}(X_{t \wedge \sigma_{A^*} }); t \geq 0 \right\}$ (here we assume $A^* > 0$), we see that
\begin{align}
\Big\{ e^{-r {(t \wedge \sigma_{A^*})}} v_{n}(X_{t \wedge \sigma_{A^*}}) - \int_0^{t \wedge \sigma_{A^*} } e^{-rs} (\mathcal{L} - r) v_{n} (X_{s-})  \diff s; \quad t \geq 0 \Big\} \label{local_martingale}
\end{align}
is a local martingale.  Here the $C^2$ ($C^1$) condition at $\{A^*,B^*\}$ for the case $X$ is of unbounded (bounded) variation can be relaxed by a version of Meyer-Ito formula as in Theorem IV.71 of \cite{ProtterBook} (see also Theorem 2.1 of  \cite{sulem}).


Suppose $\left\{T_k; k \geq 1 \right\}$ is the corresponding localizing sequence, namely,
\begin{align*}
\E^x \left[ e^{-r {(t \wedge \sigma_{A^*} \wedge T_k)}} v_n(X_{t \wedge \sigma_{A^*} \wedge T_k}) \right] &= v_n(x) + \E^x \left[  \int_0^{t \wedge \sigma_{A^*} \wedge T_k} e^{-rs}(\mathcal{L} - r) v_n (X_{s-})   \diff s \right].
\end{align*}
Now by applying the dominated convergence theorem on the left-hand side and Fatou's lemma on the right-hand side via $(\mathcal{L} - r) v_{n}(x) \leq 0$ for every $x > 0$ thanks to \eqref{generator_difference_negative} and  Lemma \ref{lemma_generator}-(2,3), we obtain
\begin{align*}
\E^x \left[ e^{-r {(t \wedge \sigma_{A^*})}} v_n(X_{t \wedge \sigma_{A^*}}) \right] \leq v_n(x) + \E^x \left[  \int_0^{t \wedge \sigma_{A^*} } e^{-rs} (\mathcal{L} - r) v_n (X_{s-})   \diff s \right].
\end{align*}
Hence  \eqref{local_martingale} is a supermartingale.

%
%
%
%

Now fix $\tau \in \S_{A^*}$. By optional sampling theorem, we have for any $M \geq 0$
\begin{align*}
&\E^x \left[ e^{-r {(\tau \wedge \sigma_{A^*} \wedge M)}} v_n(X_{\tau \wedge \sigma_{A^*} \wedge M}) \right] \\ &\leq v_n(x) + \E^x \left[  \int_0^{\tau \wedge \sigma_{A^*} \wedge M} e^{-rs}\left( (\mathcal{L} - r) v_{A^*,B^*} (X_{s-})   +  (\mathcal{L} - r) (v_n-v_{A^*,B^*}) (X_{s-}) \right) \diff s \right] \\
&\leq v_n(x) + \E^x \left[  \int_0^{\tau \wedge \sigma_{A^*} \wedge M} e^{-rs} (\mathcal{L} - r) (v_n-v_{A^*,B^*}) (X_{s-})  \diff s \right],
\end{align*}
where the last inequality holds by Lemma \ref{lemma_generator}-(2,3).  Applying the dominated convergence theorem on both sides via \eqref{expectation_bounded_generator}, we obtain the inequality:
\begin{align}
\E^x \left[ e^{-r (\tau \wedge \sigma_{A^*})} v_n(X_{\tau \wedge \sigma_{A^*}}) 1_{\{ \tau \wedge \sigma_{A^*} < \infty \}}\right] \leq v_n(x) + \E^x \left[  \int_0^{\tau \wedge \sigma_{A^*}} e^{-rs}(\mathcal{L} - r) (v_n -v_{A^*,B^*}) (X_{s-})   \diff s \right]. \label{supermtg_proof}
\end{align}
We shall take $n \rightarrow \infty$ on both sides.  For the left-hand side, the dominated convergence theorem implies
\begin{align*}
\lim_{n \rightarrow \infty} \E^x \left[ e^{-r (\tau \wedge \sigma_{A^*})} v_n(X_{\tau \wedge \sigma_{A^*}}) 1_{\{\tau \wedge \sigma_{A^*} < \infty \}} \right] =  \E^x \left[ e^{-r (\tau \wedge \sigma_{A^*})} v_{A^*,B^*}(X_{\tau \wedge \sigma_{A^*}}) 1_{\{\tau \wedge \sigma_{A^*} < \infty \}} \right].
\end{align*}
For the right-hand side, we again apply the dominated convergence theorem via \eqref{expectation_bounded_generator} to get
\begin{multline}\label{limiteqn}
\lim_{n \rightarrow \infty} \E^x \left[ \int_0^{\tau \wedge \sigma_{A^*}} e^{-rs} (\mathcal{L}-r) (v_n -v_{A^*,B^*}) (X_{s-}) \diff s\right] \\ =  \E^x \left[ \lim_{n \rightarrow \infty} \int_0^{\tau \wedge \sigma_{A^*}} e^{-rs} (\mathcal{L}-r) (v_n -v_{A^*,B^*}) (X_{s-}) \diff s\right].
\end{multline}
Now fix $\p^x$-a.e.\ $\omega \in \Omega$. By \eqref{expectation_bounded_generator} dominated convergence yields
$\lim_{n \rightarrow \infty}\int_0^{\tau (\omega)\wedge \sigma_{A^*} (\omega)} e^{-rs} (\mathcal{L}-r) (v_n -v_{A^*,B^*}) (X_{s-}  (\omega)) \diff s  = \int_0^{\tau (\omega)\wedge \sigma_{A^*} (\omega)} e^{-rs} \lim_{n \rightarrow \infty} (\mathcal{L}-r) (v_n -v_{A^*,B^*}) (X_{s-}  (\omega)) \diff s$.
Finally, since $X_{s-}(\omega) > A^*$ for Lebesgue-a.e.\ $s$ on $(0,\tau (\omega)\wedge \sigma_{A^*} (\omega))$, and  by the dominated convergence theorem,
$\lim_{n \rightarrow \infty} (\mathcal{L}-r) (v_n -v_{A^*,B^*}) (X_{s-}  (\omega)) = \int_{X_{s-}(\omega)}^\infty  \Pi(\diff u) \lim_{n \rightarrow \infty} \left( v_n(X_{s-}(\omega)- u) - v_{A^*,B^*}(X_{s-}(\omega)- u)\right) = 0$.
Hence, the limit \eqref{limiteqn} vanishes.
Therefore, by taking $n \rightarrow \infty$ in \eqref{supermtg_proof} (note $v_{A^*,B^*}(x) = v_n(x)$), we have
\begin{align*}
v_{A^*,B^*}(x)\geq \E^x \left[ e^{-r (\tau \wedge \sigma_{A^*})} v_{A^*,B^*}(X_{\tau \wedge \sigma_{A^*}}) 1_{\{\tau \wedge \sigma_{A^*} < \infty \}}\right], \quad \tau \in \S_{A^*}.
\end{align*}
This inequality and Lemma \ref{lemma_domination2}-(1) show that $v_{A^*,B^*}(x) \geq v(x; \sigma_{A^*},\tau)$ for any arbitrary $\tau \in \S_{A^*}$.

(ii)  Next, we show that $v_{A^*,B^*}(x) \leq v(x;\sigma, \tau_{B^*})$ for every $\sigma \in \S$.  Similarly to (i), we only need to focus on the set $\S_{B^*}$.
We again use $\left\{ v_n; n \geq 1 \right\}$ defined in (i).
Using the same argument as in (i), we obtain
\begin{align}\label{expectation_bounded_generator_B}
\begin{split}
\E^x \left[ \int_0^{\sigma \wedge \tau_{B^*}} e^{-rs} |(\mathcal{L}-r) (v_n - v_{A^*,B^*}) (X_{s-}) |\diff s\right] < \infty, \\
\int_0^{\sigma \wedge \tau_{B^*}} e^{-rs} |(\mathcal{L}-r) (v_n - v_{A^*,B^*}) (X_{s-}) | \diff s < \infty, \quad \p^x-a.s.,
\end{split}
\end{align}
 uniformly for any $n \geq 1$.

Because $v_n$ is not assumed to be $C^1$ nor $C^2$ at zero, we follow the approach by \cite{Loeffen_2009}.
Fix $\epsilon > 0$.  By applying Ito's formula to $\left\{ e^{-r {(t \wedge \tau_{B^*} \wedge \sigma_\epsilon)}} v_{n}(X_{t \wedge \tau_{B^*} \wedge \sigma_\epsilon}); t \geq 0 \right\}$, we see that
\begin{align}
\left\{ e^{-r {(t \wedge \tau_{B^*} \wedge \sigma_\epsilon)}} v_{n}(X_{t \wedge \tau_{B^*} \wedge \sigma_\epsilon}) - \int_0^{t \wedge \tau_{B^*} \wedge \sigma_\epsilon} e^{-rs} (\mathcal{L} - r) v_{n}(X_{s-})  \diff s; \quad t \geq 0 \right\} \label{local_martingale_B}
\end{align}
is a local martingale.  Suppose $\left\{T_k; k \geq 1 \right\}$ is the corresponding localizing sequence, we have
\begin{align*}
&\E^x \left[ e^{-r {(t \wedge \tau_{B^*} \wedge \sigma_\epsilon \wedge T_k)}} v_{n}(X_{t \wedge \tau_{B^*} \wedge \sigma_\epsilon \wedge T_k}) \right] = v_{n}(x) + \E^x \left[  \int_0^{t \wedge \tau_{B^*} \wedge \sigma_\epsilon \wedge T_k} e^{-rs} (\mathcal{L} - r) v_{n}(X_{s-})   \diff s \right] \\
&= v_n(x) + \E^x \left[  \int_0^{t \wedge \tau_{B^*} \wedge \sigma_\epsilon \wedge T_k} e^{-rs}\left( (\mathcal{L} - r) v_{A^*,B^*} (X_{s-})  + (\mathcal{L} - r) (v_n - v_{A^*,B^*})(X_{s-})  \right) \diff s \right] \\
&= v_n(x) + \E^x \left[  \int_0^{t \wedge \tau_{B^*} \wedge \sigma_\epsilon \wedge T_k} e^{-rs} (\mathcal{L} - r) v_{A^*,B^*} (X_{s-})  \diff s \right] \\ & \qquad + \E^x \left[  \int_0^{t \wedge \tau_{B^*} \wedge \sigma_\epsilon \wedge T_k} e^{-rs}(\mathcal{L} - r) (v_n - v_{A^*,B^*})(X_{s-})   \diff s \right]
\end{align*}
where we can split the expectation by \eqref{expectation_bounded_generator_B}.
Now by applying the dominated convergence theorem on the left-hand side and the monotone convergence theorem and the dominated convergence theorem respectively on the two expectations on the right-hand side (using respectively Lemma \ref{lemma_generator}-(1,2) and \eqref{expectation_bounded_generator_B}), we obtain
\begin{align*}
\E^x \left[ e^{-r {(t \wedge \tau_{B^*} \wedge \sigma_\epsilon)}} v_{n}(X_{t \wedge \tau_{B^*} \wedge \sigma_\epsilon}) \right] &= v_{n}(x) + \E^x \left[  \int_0^{t \wedge \tau_{B^*} \wedge \sigma_\epsilon } e^{-rs} (\mathcal{L} - r) v_{n}(X_{s-})  \diff s \right].
\end{align*}
Hence  \eqref{local_martingale_B} is a martingale.

Now fix $\sigma \in \S_{B^*}$. By the optional sampling theorem, we have for any $M \geq 0$ using Lemma \ref{lemma_generator}-(1,2)
\begin{align*}
&\E^x \left[ e^{-r {(\sigma \wedge \tau_{B^*} \wedge \sigma_\epsilon \wedge M)}}  v_{n}(X_{\sigma \wedge \tau_{B^*} \wedge \sigma_\epsilon \wedge M}) \right]  =  v_{n}(x) + \E^x \left[ \int_0^{\sigma \wedge \tau_{B^*} \wedge \sigma_\epsilon \wedge M} e^{-rs} (\mathcal{L}-r) v_n(X_{s-})\diff s\right]\\
 &\geq v_n(x) + \E^x \left[  \int_0^{\sigma \wedge \tau_{B^*} \wedge \sigma_\epsilon \wedge M} e^{-rs} (\mathcal{L} - r) (v_n-v_{A^*,B^*}) (X_{s-})   \diff s \right].
\end{align*}
Applying the dominated convergence theorem on both sides by \eqref{expectation_bounded_generator_B}, we have
\begin{align*}
&\E^x \left[ e^{-r {(\sigma \wedge \tau_{B^*} \wedge \sigma_\epsilon)}}  v_{n}(X_{\sigma \wedge \tau_{B^*} \wedge \sigma_\epsilon }) 1_{\{\sigma \wedge \tau_{B^*} \wedge \sigma_\epsilon < \infty \}} \right] \geq  v_{n}(x) + \E^x \left[  \int_0^{\sigma \wedge \tau_{B^*} \wedge \sigma_\epsilon} e^{-rs} (\mathcal{L} - r) (v_n-v_{A^*,B^*}) (X_{s-})  \diff s \right].
\end{align*}
Because $\sigma_\epsilon \rightarrow \sigma_0$ ($\tau_{B^*} \wedge \sigma_\epsilon \rightarrow \tau_{B^*}$ )  a.s., the bounded convergence theorem yields
\begin{align*}
&\E^x \left[ e^{-r {(\sigma \wedge \tau_{B^*})}}  v_{n}(X_{\sigma \wedge \tau_{B^*}}) 1_{\{ \sigma \wedge \tau_{B^*} < \infty \}}\right] \geq  v_{n}(x) + \E^x \left[  \int_0^{\sigma \wedge \tau_{B^*}} e^{-rs} (\mathcal{L} - r) (v_n-v_{A^*,B^*}) (X_{s-}) \diff s \right].
\end{align*}
Finally, we can take $n \rightarrow \infty$ on both sides along the same line as in (i) and we obtain
\begin{align*}
v_{A^*,B^*}(x) &\leq  \E^x \left[ e^{-r {(\sigma \wedge \tau_{B^*})}} \lim_{n \rightarrow \infty} v_{n}(X_{\sigma \wedge \tau_{B^*}}) 1_{\{ \sigma \wedge \tau_{B^*} < \infty \}}\right] \\  &= \E^x \left[ e^{-r {(\sigma \wedge \tau_{B^*})}} (v_{A^*, B^*}(X_{\sigma \wedge \tau_{B^*} }) 1_{\{ X_{\sigma \wedge \tau_{B^*}} \neq 0 \}} + v_{A^*, B^*}(0+) 1_{\{ X_{\sigma \wedge \tau_{B^*}} = 0  \}} ) 1_{\{ \sigma \wedge \tau_{B^*} < \infty  \}}\right] \\
&\leq \E^x \left[ e^{-r {(\sigma \wedge \tau_{B^*})}} v_{A^*, B^*}(X_{\sigma \wedge \tau_{B^*} }) 1_{\{ \sigma \wedge \tau_{B^*} < \infty \}}  \right].
\end{align*}
%
This together with Lemma \ref{lemma_domination2}-(2) shows that  $v_{A^*,B^*}(x) \leq v(x; \sigma,\tau_{B^*})$ for any arbitrary $\sigma \in \S_{B^*}$.  
\end{proof}

\begin{proof}[Proof of Theorem \ref{theorem_equilibrium_zero}]
When $\nu = 0$, then the same results as (i) of the proof of Theorem \ref{theorem_equilibrium}  hold by replacing $A^*$ with $0$ and $\tau_{A^*}$
 with ${\sigma_0}$.  Now suppose $\nu > 0$. Using the same argument as in the proof of Theorem \ref{theorem_equilibrium} with $\tau_{A^*}$ replaced with ${\sigma_0}$ and the argument with $\sigma_\epsilon$ as in (ii) of the proof of Theorem \ref{theorem_equilibrium}, the supermartingale property of $\left\{ e^{-r (t \wedge {\sigma_0})}v_{0+,B^*}(X_{t \wedge {\sigma_0}}); t \geq 0 \right\}$ holds.  This together with Lemma \ref{lemma_domination2}-(1) shows, for any $\tau \in \S$,
\begin{align*}
v_{0+,B^*}(x) \geq \E^x \left[ e^{-r \tau} v_{0+,B^*}(X_\tau) 1_{\{\tau < \infty\}}\right] \geq \E^x \left[ e^{-r \tau} (h(X_\tau) 1_{\{\tau < {\sigma_0}\}} - (\acheck - \gamma_s) 1_{\{X_\tau = 0\}}) 1_{\{\tau < \infty \}}\right] = v(x; \sigma_{0+}, \tau).
\end{align*}
 As in the proof of Lemma \ref{lemma_generator}-(2), $\left\{ e^{-r (t \wedge \tau_{B^*})}v_{0+,B^*}(X_{t \wedge \tau_{B^*}}); t \geq 0 \right\}$ is a martingale.  This together with Lemma \ref{lemma_domination2}-(2) shows that $v_{0+,B^*}(x) \leq v(x;\sigma, \tau_{B^*})$ for all $\sigma \in \S_{B^*}$.  
\end{proof}

\begin{small}\textbf{Acknowledgements.} This work is supported by NSF Grant DMS-0908295, Grant-in-Aid for Young Scientists (B) No.\ 22710143, the Ministry of Education, Culture, Sports, Science and Technology, and Grant-in-Aid for Scientific Research (B) No.\ 23310103, No.\ 22330098, and (C) No.\ 20530340, Japan Society for the Promotion of Science. We thank two anonymous referees for their thorough reviews and insightful comments that help improve the presentation of this paper.\end{small}

\bibliographystyle{abbrv}
\bibliographystyle{apalike}

\bibliographystyle{agsm}\begin{small}
\bibliography{bib_game}\end{small}
\end{document}